\documentclass[letter,11pt]{article}
\usepackage[letterpaper,text={6.6in,9.3in}]{geometry}

\usepackage{algorithm}
\usepackage[noend]{algpseudocode}

\usepackage{authblk}
\usepackage{hyperref}
\usepackage{amsfonts}
\usepackage{amsmath}
\usepackage{amssymb}
\usepackage{amsthm}
\usepackage{graphicx}
\usepackage{color}
\usepackage[none]{hyphenat}
\usepackage{wrapfig}
\usepackage[justification=centering]{caption}
\usepackage{bm}
\usepackage{enumitem}
\usepackage{multirow}
\usepackage{comment}
\usepackage{float}

\theoremstyle{plain}
\newtheorem{theorem}{Theorem}[section]

\newtheorem*{claim}{Claim}
\newtheorem{corollary}[theorem]{Corollary}
\theoremstyle{definition}
\newtheorem{definition}[theorem]{Definition}

\newtheorem*{remark}{Remark}

\newcommand{\fix}[1]{\textbf{\textcolor{red}{#1}}}
\newcommand{\prob}{\text{Pr}}
\newcommand{\eps}{\epsilon}
\newcommand{\ev}{\mathbb{E}}

\theoremstyle{plain}
\theoremstyle{definition}

\begin{document}
	
\pagenumbering{gobble}
\begin{titlepage}	
	{ 
		\title{The Role of A-priori Information in Networks of Rational Agents
			{\let\thefootnote\relax\footnote{
					This work is partially supported by a grant from the
					Blavatnik Cyber Security Council,
					the Blavatnik Family Computer Science Research Fund
					and the KIPP Foundation.
			}}
		}
		\author{Yehuda Afek}
		\author{Yishay Mansour}
		\author{Shaked Rafaeli}
		\author{Moshe Sulamy}
		\affil{Tel-Aviv University}
		\date{}
		\maketitle
	}
	\begin{abstract}

Until now, distributed algorithms for rational agents have assumed a-priori knowledge of $n$,
the size of the network.
This assumption is challenged here by proving how much a-priori knowledge is necessary for equilibrium in different distributed computing problems.
Duplication - pretending to be more than one agent - is the main tool used by agents to deviate and increase their utility when not enough knowledge about
$n$ is given.
The a-priori knowledge of $n$ is formalized as a Bayesian setting where at the beginning of the algorithm agents only know a prior $\sigma$, a distribution from which they know $n$ originates.

We begin by providing new algorithms for the Knowledge Sharing and Coloring problems when $n$ is a-priori known to all agents.  
We then prove that when agents have no a-priori knowledge of $n$, i.e.,
	the support for $\sigma$ is infinite,
	equilibrium is impossible for the Knowledge Sharing problem.
Finally, we consider priors with finite support and
find bounds on the necessary interval $[\alpha,\beta]$ that contains the support of $\sigma$,
i.e., $\alpha \leq n \leq \beta$,
for which we have an equilibrium.
When possible, we extend these bounds to hold for \emph{any} possible protocol.

\footnotetext[1]{
	This paper is the full version of the DISC 2018 \cite{afek_et_al:LIPIcs:2018:9794} paper.}

	\end{abstract}

	\hfill
\end{titlepage}
\pagenumbering{arabic}

\section{Introduction}
\label{section_intro}

The complexity and simplicity of most distributed computing problems
depend on the inherent a-priori knowledge given to all participants.
Usually, the more information processors in a network start with,
the more efficient and simple the algorithm for a problem is.
Sometimes, this information renders an otherwise unsolvable problem, solvable.

In game-theoretic distributed computing, algorithms run in a network of \emph{rational agents} that may deviate from an algorithm if they deem the deviation more profitable for them.
Rational agents have always been assumed to know the number of participants in the network \cite{Abraham:2011,DISC13/ADH,Afek:2014:DCB:2611462.2611481,
	Halpern:2016:RCE:2933057.2933088,YifrachLeader},
when in fact this assumption is not only unrealistic in today's internet, but also provides agents with
non-trivial information which is critical for equilibrium.

Consider for example
a large world-wide social network on which a distributed algorithm
between a large portion of its members is run.
It does not necessarily have the time to verify the number of participants,
or the service it provides with the algorithm will be irrelevantly slow. If $n$ is known to all participants, as was assumed in previous works about rational agents, that would not be a problem. However, what if $n$ is not known beforehand, allowing one of the participants to skew the result in his favor?

The problems we examine here can be solved in the game-theoretic setting when $n$ is a-priori known.
However, learning the size of the network reliably is not possible with rational agents and thus we show that some a-priori knowledge of $n$ is critical for equilibrium.
That is, without any knowledge of $n$, equilibrium for some problems is impossible.
In contrast, these problems can be solved without knowledge of $n$
if the participants are not rational,
since we can acquire the size of the network using broadcast and echo.

When $n$ is not a-priori known, 
agents may deviate from the algorithm by \emph{duplicating} themselves to affect the outcome.
This deviation is also known as a Sybil Attack \cite{Douceur:2002:SA:646334.687813},
commonly used to manipulate internet polls,
increase page rankings in Google \cite{Bianchini:2005:IP:1052934.1052938}
and affect reputation systems such as eBay \cite{Bhattacharjee:2005:ABS:1080192.1080203, Cheng:2005:SRM:1080192.1080202}.  In this paper, we use a Sybil Attack as a method for agents to skew protocols in their favor and increase their utility.
For each problem presented,
an equilibrium when $n$ is known is provided or was provided in a previous work,
thus for these problems, agents that do not duplicate cannot increase their utility.
Obviously, deviations from the algorithm that include both duplicating and additional cheating are also possible.

Intuitively, the more agents an agent is disguised as,
the more power to affect the output of the algorithm it has.
For every problem, we strive to find the maximum number of duplications a cheater may be allowed to duplicate without gaining an advantage when compared to following the protocol legally, i.e., the maximum number of duplications for which equilibrium is still possible.
This maximum number of duplications depends on whether other agents will detect that a duplication has taken place,
since the network could not possibly be this large.
To detect this situation they need to possess some knowledge about the network size, or about a specific structure.

Agents may possess partial a-priori knowledge of $n$, i.e., they do not know $n$ precisely but instead have a prior belief on the possible values of $n$. We formalize this notion by using a Bayesian setting in which agents a-priori know that $n \sim \sigma$ for some discrete distribution $\sigma$ over $\mathbb{N}$. 
We use this Bayesian setting to determine the requirements necessary of agents' a-priori knowledge of $n$ for equilibria to be possible. More specifically, we prove impossibility of equilibria in some settings, and find bounds on $\sigma$ such that if agents a-priori know that $n$ originates from $\sigma$, i.e., $n \sim \sigma$, an equilibrium exists. In the problems we examine, the important characteristic of $\sigma$ is its support: whether the support is finite and the size of the interval $[\alpha,\beta]$ that contains the support of $\sigma$ from which $n$ is drawn.
These bounds hold for both deterministic and non-deterministic algorithms.

Using these bounds, we show what algorithms may be used in specific networks.
For example, in an internal business network,
some algorithms may work because every member in the network knows
there are no more than several thousand computers in the network,
while for other algorithms this knowledge is not tight enough.

Table \ref{tbl:results} summarizes our contributions and related previous work (where there is a citation).
In every row, a different distributed computing problem is examined in different settings of a-priori knowledge.
Known $n$ refers to the case where all agents in the network start the algorithm when they \emph{know} $n$. 
Unknown $n$ refers to the Bayesian environment where agents know a prior $\sigma$ from which $n$ is drawn, but $\sigma$ is a distribution with infinite support, namely, there is some $n_0$ such that for any $m > n_0$, $\prob[n = m] > 0$. 
The two rightmost columns refer to two types of priors with finite support: A uniform distribution on integers in an interval $[\alpha,\beta]$, and a geometrically decreasing distribution on integers in an interval $[\alpha,\beta]$.

\begin{table}[ht]
	\centering
	\begin{tabular}{|c|c|c|c|c|}
		\hline
		Problem & Known $n$ & Unknown $n$ & uniform $\sigma \subseteq [\alpha,\beta]$  & geometric $\sigma \subseteq [\alpha,\beta]$ \\ \hline
		Coloring
		& \shortstack{ \checkmark \\ Section~\ref{section_alg}}
		& \shortstack{Open Question}
		& \shortstack{Equilibrium \\ for any $\alpha,\beta$ *}
		& \shortstack{Equilibrium \\ when $\beta \leq  2\alpha - 1 $* }
		\\ \hline
		\shortstack{Leader Election\\\space}
		& \shortstack{\checkmark \\ ADH'13 \cite{DISC13/ADH}}
		& \shortstack{Impossible \\ ADH'13 \cite{DISC13/ADH}}
		& \shortstack{Equilibrium iff \\
					  $\beta \leq \alpha + 1$}
		& \shortstack{Equilibrium iff \\
			$\beta = \alpha$}
		\\ \hline
		Knowledge Sharing
		& \multirow{2}{*}{\shortstack{\checkmark\\ AGLS'14 \cite{Afek:2014:DCB:2611462.2611481}}}
		& \multirow{2}{*}{\shortstack{Impossible \\ Section~\ref{section_imp}}}
		& \shortstack{ Equilibrium when \\ $\beta \leq 2\alpha - 1$}
		& \shortstack{ Equilibrium when \\ $\beta \leq 2\alpha - 1$}
		\\ \cline{1-1} \cline{5-5} \cline{4-4}
		$2$-Knowledge Sharing
		&
		&
		& \shortstack{ Equilibrium \\ For any $\alpha,\beta$}
		& \shortstack{ Equilibrium \\ when $\beta \leq 2\alpha - 1$}
		\\ \hline
		Partition, Orientation
		& \shortstack{ \checkmark \\ Section~\ref{section_approx}}
		&  \shortstack{ \checkmark \\ Section~\ref{section_approx}}
		& \shortstack{$\forall \sigma$ an equilibrium \\ exists}
		& \shortstack{$\forall \sigma$ an equilibrium \\ exists}
		\\ \hline
	\end{tabular}
		\caption{Summary of paper contributions,
		equilibria and impossibility results for different problems with different a-priori knowledge about $n$
		\\\textbf{*} Results hold for ring networks
		\\\textbf{**} See Section~\ref{section_approx} for full analysis for Knowledge Sharing}
	\label{tbl:results}
\end{table}


\subsection{Related Work}

The connection between distributed computing and game theory
stemmed from the problem of secret sharing \cite{CACM/Shamir79}.
Further works continued the research on secret sharing and multiparty computation
when both Byzantine and rational agents are present
\cite{PODC/AbrahamDGH06, PODC/DaniMRS11, TCC/FuchsbauerKN10, SCN/GordonK06,  ICALP/GroceKTZ12, CYPTO/LysyanskayaT06}. 

Another line of research presented the BAR model (Byzantine, acquiescent and rational) \cite{SOSP/AiyerACDMP05, PODC/MoscibrodaSW06, OPODIS/WongLACD11},
while a related line of research discusses converting solutions with a mediator to
cheap talk \cite{PODC/AbrahamDGH06, TCC/AbrahamDH08, Barany1992, JET/Ben-Porath03,  CRYPTO/DodisHR00,
	PODC/LepinskiMP04, TARK/McGrewPS03, TCS/Shoham05, ECON/UAV02, ET/UA04}.

Abraham, Dolev, and Halpern \cite{DISC13/ADH} were the first to present
protocols for networks of rational agents,
specifically protocols for Leader Election.
In \cite{Afek:2014:DCB:2611462.2611481} the authors continue this line
of research by providing basic building blocks for game theoretic distributed algorithms,
namely a wake-up and knowledge sharing equilibrium building blocks.
Algorithms for consensus, renaming, and leader election are presented using these building blocks.
Consensus was researched further by Halpern and Vilacça \cite{Halpern:2016:RCE:2933057.2933088},
who showed that there is no ex-post Nash equilibrium, and a Nash equilibrium that tolerates $f$ failures
under some minimal assumptions on the failure pattern.
Yifrach and Mansour \cite{YifrachLeader} studied fair Leader Election protocols,
giving an almost tight resilience analysis.
Bank, Sulamy, and Waserman \cite{gt_idspace} examined the case where the $id$ space is limited,
calculating the minimal threshold for equilibrium.

Coloring and Knowledge Sharing have been studied extensively in a distributed setting
\cite{AttiyaBook,Awerbuch:1989:NDL:1398514.1398717,Cole:1986:DCT:10366.10368,
	Kuhn:2006:CDG:1146381.1146387,Linial1986,Linial:1987:DGA:1382440.1382990,
	Szegedy:1993:LBG:167088.167156}.
An algorithm for Knowledge Sharing  with rational agents was presented in \cite{Afek:2014:DCB:2611462.2611481}, while
Coloring with rational agents has not been studied previously, to the best of our knowledge.

\section{Model}
\label{section_model}
We use the standard message-passing, synchronous model, where the network is a bidirectional graph
$G=\left(V,E\right)$ with $n$ nodes, each node representing an agent with unlimited computational power,
and $|E|$ edges over which the agents communicate in rounds.
$G$ is assumed to be $2$-vertex-connected\footnotemark.
Throughout the entire paper, $n$ always denotes the \emph{actual} number of nodes in the network. In Section~\ref{section_alg} it is the exact size of the network. In Section~\ref{section_imp} and Section~\ref{section_approx}, agents treat it as a random variable drawn from a prior $\sigma$.
\footnotetext{
	This property was shown necessary in \cite{Afek:2014:DCB:2611462.2611481}, since if a bottleneck node exists it can alter any message passing through it.
	Such a deviation cannot be detected since all messages between the sub-graphs this node connects
	must traverse through it.
	This node can then skew the algorithm according to its preferences.
}

Initially, each agent knows its \emph{id} and input,
but not the \emph{id} or input of any other agent. 
The agents in the network have a prior over the information they do not know. For any problem, we demand the set of private information agents have is finite, and so we assume that the prior agents have over other agents' private information is uniform. The exception for this assumption is $n$, the size of the network: agents know $n$ precisely or know $\sigma$, an explicitly stated prior.

We assume all agents start the protocol together at round $0$,
i.e., all agents wake up at the same time.
If not, we can use the Wake-Up \cite{Afek:2014:DCB:2611462.2611481} building block to relax this assumption.

\subsection{Equilibrium in Distributed Algorithms}
\label{section_legal_correct}

Informally, a distributed algorithm is an equilibrium if no agent at no point in the execution can do better by unilaterally deviating from the algorithm. 
When considering a deviation, an agent assumes all other agents follow the algorithm,
i.e., it is the only agent deviating.


Formally, let $o_a$ be the output of agent $a$, let $\Theta$ be the set of all possible output vectors,
and denote the output vector $O =(o_1,\dots,o_n) \in \Theta$, where $O[a]=o_a$.
Let $\Theta_L$ be the set of \emph{legal} output vectors, in which the protocol terminates successfully,
and let $\Theta_E$ be the set of \emph{erroneous} output vectors, such that
$\Theta = \Theta_L \cup \Theta_E$ and $\Theta_L \cap \Theta_E = \varnothing$.

Each agent $a$ has a utility function $u_a: \Theta \rightarrow \mathbb{N}$.
The higher the value assigned by $u_a$ to an output vector,
the better this vector is for $a$.
As in previous works \cite{DISC13/ADH,Afek:2014:DCB:2611462.2611481,YifrachLeader},
the utility function is required to satisfy the \emph{Solution Preference} property
which guarantees that an agent never has an incentive to fail the algorithm.
Otherwise, they would simply be Byzantine faults.
An agent fails the algorithm only when it detects that another agent had deviated.
                                                                                     
\begin{definition}[Solution Preference]                                              
	The utility function $u_a$ of an agent $a$ never assigns a higher utility           
	to an erroneous output than to a legal one, i.e.:                                   
	$$\forall{a,O_L \in \Theta_L,O_E \in \Theta_E}: u_a(O_L) \geq u_a(O_E)$$
\end{definition}

We differentiate the \emph{legal} output vectors, which ensure the output is valid and not erroneous, from the \emph{correct} output vectors, which are output vectors that are a result of a correct execution of the algorithm, i.e., without any deviation.
Solution Preference guarantees agents never prefer an erroneous output. However, they may prefer a \emph{legal} but \emph{incorrect} output.

The Solution Preference property introduces the threat agents face when deviating:
Agents know that if another agent catches them cheating,
it outputs $\bot$ and the algorithm fails.
In other words, the output is erroneous, i.e., in $\Theta_E$.

For simplicity, we assume agents only have preferences over their own output,
i.e., for any $O_1,O_2 \in \Theta_L$ in which $O_1[a]=O_2[a]$, $u_a(O_1) = u_a(O_2)$.
Additionally, each agent $a$ has a \emph{single} preferred output value $p_a$,
and we normalize the utility function values, such that\footnotemark:
\begin{equation}
\label{eq_u_a}
u_a(O)=\begin{cases}
	1 & o_a=p_a$ and $O \in \Theta_L
	\\ 0 & o_a \neq p_a$ or $O \in \Theta_E\end{cases}
\end{equation}

These assumptions are for convenience only and can easily be removed.
Our results hold for \emph{any} utility function that satisfies Solution Preference.

\footnotetext{
	This is the weakest assumption since it still leaves a cheating agent
	with the highest incentive to deviate, while still satisfying Solution Preference.
	A utility assigning a lower value for failure than $o_a \neq p_a$
	would deter a cheating agent from deviating.
}

\begin{definition} [Expected Utility]
Let $r$ be a round in a specific execution of an algorithm. Let $a$ be an arbitrary agent. For each possible output vector $O$,
let $x_O(s,r)$ be the probability, estimated by agent $a$ at round $r$,
that $O$ is output by the algorithm if $a$ takes step $s$~\footnotemark,
and all other agents follow the algorithm.
The \emph{Expected Utility} $a$ estimates for step
$s$ in round $r$ of that specific execution is:
$$ \mathbb{E}_{s,r} [u_a]=\sum\limits_{O \in \Theta} x_O(s,r) \cdot u_a(O) $$
\end{definition}

We can also consider the expected utility from following a \emph{strategy}. Let $S_a$ denote a series of steps for agent $a$ beginning at round $r$ until the end of the protocol, i.e., a list of actions $a$ performs every round from $r$ and on. These may specify different reactions to situations during the course of the protocol. Then we can consider all possible outputs where $a$ acts according to $S_a$ and other agents act according to the protocol, and denote the expected utility over all these outputs as $\ev_{S_a,r}[u_a]$.

\footnotetext{
	A step specifies the entire operation of the agent in a round.
	This may include drawing a random number, performing any internal computation,
	and the contents and timing of any message delivery.
}

An agent will deviate whenever the deviating step has a strictly
higher expected utility than the expected utility of the next step of the algorithm,
even if that deviating step also increases the risk of an erroneous output.

Let $\Lambda$ be an algorithm.
If by deviating from $\Lambda$ and taking step $s$,
the expected utility of $a$ is higher, we say that agent $a$ has an \emph{incentive to deviate} (i.e., cheat).
For example, at round $r$ algorithm $\Lambda$ may dictate that $a$ flips a fair binary coin
and sends the result to all of its neighbors.
Any other action by $a$ is considered a \emph{deviation}:
whether the message was not sent to all neighbors,
sent later than it should have, or whether the coin toss was not fair,
e.g., $a$ only sends $0$ instead of a random value. 
If no agent can unilaterally increase its expected utility by deviating from $\Lambda$,
we say that the protocol is an \emph{equilibrium}.
Equilibrium is defined over a \emph{single} deviating agent,
i.e., there are no coalitions of agents.

\begin{definition}[Distributed Equilibrium]
	Let $s(r)$ denote the next step of algorithm $\Lambda$ in round $r$.
	$\Lambda$ is an equilibrium if for any deviating step $\bar{s}$,
	at any round $r$ of every possible execution of $\Lambda$ where all steps by all agents up to round $r$ were according to $\Lambda$:
	$$\forall{a,r,\bar{s}}: \mathbb{E}_{s(r),r}[u_a] \geq \mathbb{E}_{\bar{s},r}[u_a]$$
\end{definition}

Distributed Equilibrium is a sort of truthful equilibrium: a protocol where if agents assume all other agents follow the protocol truthfully, their best action is to truthfully follow the protocol as well.

It is important to emphasize that for any non-trivial distributed algorithm,
the outcome cannot be calculated using only private data, without communication. 
For rational agents, no agent can calculate the output privately at the beginning of the algorithm,
since if it could calculate the output and know that its resulting utility will be $0$,
it would surely lie over its initial information to avoid losing, preventing equilibrium.
If it knows its resulting utility is $1$,
it has no incentive to cheat. But then there exists an instance of the same problem where the agents has a different preference over the output, making that protocol not an equilibrium.
Formally, 
for any agent $a$ and any step $s$ of the agent that does not necessarily result in algorithm failure,
it must hold that: $\mathbb{E}_{s,0}[u_a] \notin \{0,1\}$
(a value of $0$ means an agent will surely not get its preference,
and $1$ means it is guaranteed to get its preference).

\subsection{Priors on $n$}
In Section~\ref{section_imp} and Section~\ref{section_approx} agents are in a Bayesian environment where they do not know $n$, the size of the network, but instead know a prior $\sigma$ over $\mathbb{N}$ from which $n$ originates. The support of $\sigma$ is:
\[
	supp(\sigma) = \{x | \prob[n=x] > 0\} 
\]

\begin{definition} [Infinite Support]
	A prior $\sigma$ agents have on $n$ has \emph{infinite support} if there is some network size $n_0$ that any network size larger than $n_0$ is possible with non-zero probability. Formally, 
	a prior $\sigma$ on $n$ has an infinite support if there exists $n_0 \in \mathbb{N}$ s.t.:
	\[
		\forall x > n_0 : \prob[n=x] > 0
	\]
\end{definition}

\begin{definition} [Finite support]
	A prior $\sigma$ has \emph{finite support in $[\alpha,\beta]$} if there exists an interval $[\alpha,\beta]$ over $\mathbb{N}$ for which $supp(\sigma) \subseteq [\alpha,\beta]$. Denoted in short $\sigma \subseteq [\alpha,\beta]$.
\end{definition}

\begin{definition}[Uniform Prior]
	A uniform distribution over integers in an interval $[\alpha,\beta]$ is defined using the following probability mass function:
	\[
		F(t) = \prob[n \leq t ] = \frac{t-\alpha+1}{\beta-\alpha+1}
	\]
\end{definition}

\begin{definition}[Geometric Prior]
	As a decreasing geometric distribution we use a factor $\frac{1}{2}$ geometric distribution starting at $\alpha$ and decreasing until $\beta$.
	Formally, let $n$ be a random variable drawn from $\sigma$. $\sigma$ is a geometric $\frac{1}{2}$ distribution starting at $\alpha$ until $\beta$, where the tail above $\beta$ is spread uniformly along $[\alpha,\beta]$. The "tail" is:
	\[
	\prob[n > \beta] = 1 - \prob[n \leq \beta] = 1 - (1-(\frac{1}{2})^{\beta-\alpha+1})=2^{\alpha-1-\beta}
	\]
	Thus for any $t \geq \alpha$:
	\[
	f(t) = \prob[n = t ] = 2^{-t+\alpha-1} + \frac{2^{\alpha - \beta-1}}{\beta - \alpha + 1} = 2^{\alpha-t-1} + c
	\]
	\[
	F(t) = \prob[n \leq t] = 
	\sum_{j = \alpha}^{t}  (  2^{\alpha-j-1} 
	+ c )
	= \big( 1-2^{\alpha-t-1} \big) + (t-\alpha +1)c
	\]
	
	Where $c =  \frac{2^{\alpha - \beta-1}}{\beta - \alpha + 1}$.
\end{definition}

\subsection{Knowledge Sharing}
The Knowledge Sharing problem (adapted from \cite{Afek:2014:DCB:2611462.2611481}) is defined as follows:
\begin{enumerate}
	
	\item Each agent $a$ has a private input $i_a$, in addition to its $id$,
	and a function $q$, where $q$ is identical at all agents.  Agents know the possible output space of $q$ before the algorithm begins.
	
	\item A Knowledge Sharing protocol terminates \emph{legally} if all agents
	output the \emph{same} value, i.e., $\forall{a,b}: o_a=o_b \neq \bot$. 
	Thus the set $\Theta_L$ is defined as:
	$ O \in \Theta_L \iff \forall{a,b}: O(a) = O(b) \neq \bot$.
	
	\item A Knowledge Sharing protocol terminates \emph{correctly} (as described in Section~\ref{section_legal_correct}) if
	each agent outputs the value $q(I)$
	over the input values $I=\{i_1,\dots,i_n\}$ of all other agents\footnotemark.
	\footnotetext{Notice that any output is legal as long as it is the output of all agents,
	but only a single output value is considered \emph{correct} for a given input vector.}

	\item The function $q$ satisfies the Full Knowledge property:
	\begin{definition}[Full Knowledge Property]
		\label{full_knowledge}
		A function $q$ fulfills the \emph{full knowledge} property if, 
		for each agent that does not know the input values of all other agents,
		any output in the range of $q$ is \emph{equally} possible.
		Formally, for any $1 \leq j \leq m$, fix $(x_1,\dots,x_{j-1},x_{j+1},\dots,x_m)$
		and denote $z_y = |\{x_j | q(x_1,\dots,x_j,\dots,x_m)=y\}|$.
		A function $q$ fulfills the \emph{full knowledge} property if, for any possible output $y$ in the range of $q$,
		$z_y$ is the same\footnotemark.
		
		\footnotetext{The definition assumes input values are drawn uniformly, otherwise the definition
			of $z_y$ can be expanded to the sum of probabilities over every input value for $x_j$.}
	\end{definition}
\end{enumerate}

 We differentiate two variants of Knowledge Sharing: 
	\begin{itemize}
		\item $k$-Knowledge Sharing - where $|range(q)| = k$ and $k$ is known to all agents at the beginning of the algorithm. For example, in $2$-Knowledge Sharing the output space is $\{0,1\}$ and all agents know the possible outputs of the protocol are from $\{0,1\}$. If every agent's input is a random bit, a $2$-Knowledge Sharing protocol performs a shared coin flip.
		\item Knowledge Sharing - A protocol for Knowledge Sharing is a protocol that solves $k$-Knowledge Sharing for any possible $k$. In the same manner, a protocol is an \emph{equilibrium} if for any possible $k \in \mathbb{N}$, no agent has an incentive to cheat.
	\end{itemize}

We assume that each agent $a$ prefers a certain output value $p_a$.

\subsection{Coloring}
In the Coloring problem \cite{Cole:1986:DCT:10366.10368,Linial1986},
$\Theta_L$ is any $O$ such that
$\forall{a} : o_a \neq \bot$ and  $\forall{(a,b) \in E} : o_a \neq o_b$.
We assume that every agent $a$ prefers a specific color $p_a$.

\subsection{Sybil Attacks}

Let $a$ be a malicious agent with $\delta$ outgoing edges.
A possible deviation for $a$ is to simulate imaginary agents $a_1$, $a_2$
and to answer over some of its edges as $a_1$, and over the others as $a_2$,
as illustrated in Figure~\ref{figure:duplicated_node}.
From this point on $a$ acts as if it is 2 agents. 
We assume that the \emph{id}~space is much larger than $n$, 
allowing us to disregard the probability that the fake \emph{id} collides with an existing \emph{id}, an issue dealt with in \cite{gt_idspace}.

\begin{figure}[H]
	\caption{Agent $a$ acting as separate agents $a_1$, $a_2$}	
	\label {figure:duplicated_node}
	\captionsetup{justification=centering}
	\includegraphics[scale=0.4]{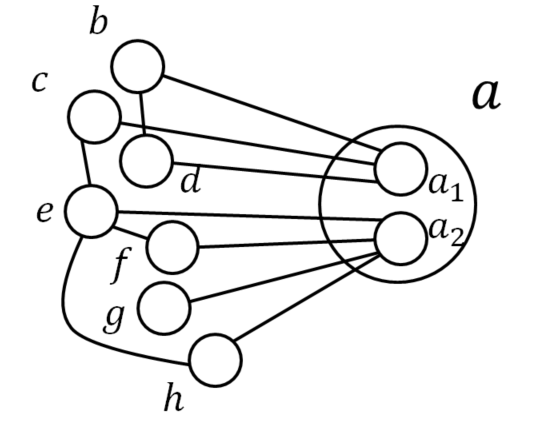}
	\centering
\end{figure}

Regarding the output vector, notice that an agent that pretends to be more than one agent
still outputs a \emph{single} output at the end.
The duplication causes agents to execute the algorithm as if it is executed on a graph $G'$ (with the duplicated agents) instead of the original graph $G$;
however, the output is considered legal if $O = (o_a,o_b,\dots) \in \Theta_L$
rather than if $(o_{a_1},o_{a_2},o_b,\dots) \in \Theta_L$.

\section{Algorithms}
\label{section_alg}

Here we present algorithms for Knowledge Sharing (Section~\ref{section:alg:ks})
and Coloring (Section~\ref{section:alg:coloring}).

The Knowledge Sharing algorithm presented here is an equilibrium in a ring network
when no cheating agent pretends to be more than $n$ agents, improving the Knowledge Sharing algorithm in \cite{Afek:2014:DCB:2611462.2611481}.
The Coloring algorithms are equilibria in any $2$-connected graph
when agents a-priori know $n$.


Notice that using an algorithm as a subroutine is not trivial in this setting,
even if the algorithm is an equilibrium,
as the new context as a subroutine may allow agents to deviate
towards a different objective than was originally proven.
Thus, whenever a subroutine is used, the fact that it is an equilibrium should be proven.

\subsection{Knowledge Sharing in a Ring}
\label{section:alg:ks}

First we describe the \texttt{Secret-Transmit($i_a$,$r$,$b$)} building block in which
agent $b$ learns the input $i_a$ of agent $a$ at round $r$,
and no other agent in the ring learns any information about this input.
To achieve this, agent $a$ selects a random number $R_a$,
and let $X_a=R_a \oplus i_a$.
It sends $R_a$ clockwise and $X_a$ counter-clockwise until each reaches the agent before $b$.
At round $r-1$, these neighbors of $b$ simultaneously send $b$ the values $X_a$ and $R_a$,
thus $b$ receives the information at round $r$.

\begin{figure}[H]
	\caption{Example of \texttt{Secret-Transmit$(i_{a_1}, r, a_3)$}}	
	\label {figure:secrettransmitexample}
	\captionsetup{justification=centering}
	\includegraphics[scale=0.35]{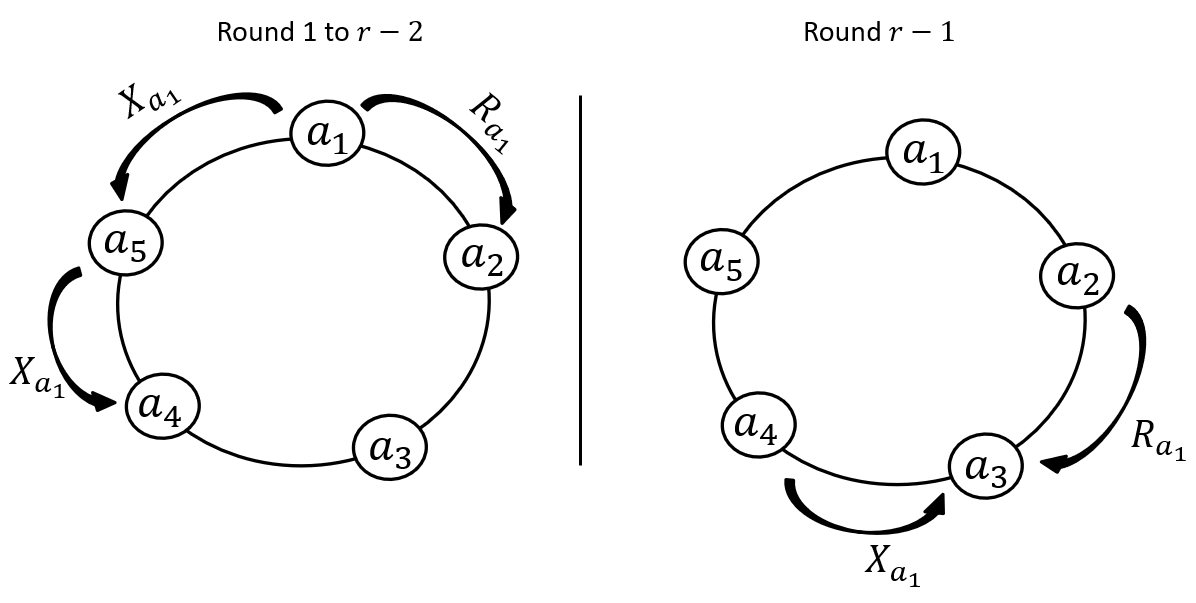}
	\centering
\end{figure}

We assume a global orientation around the ring.
This assumption can easily be relaxed via Leader Election \cite{Afek:2014:DCB:2611462.2611481},
which is an equilibrium in this application since the orientation
has no effect on the output.
The algorithm works as follows:

\begin{algorithm}[H]
	\caption{Knowledge Sharing in a Ring}
	\label{alg_ks}
	\begin{algorithmic}[1]
		\item All agents execute Wake-Up \cite{Afek:2014:DCB:2611462.2611481} to learn the ids
		of all agents and $n'$, the size of the ring (which may include duplications)
		\State For each agent $a$, denote $b^1_a$ the clockwise neighbor of $a$,
		and $b^2_a$ the agent at distance $\lfloor \frac{n'}{2} \rfloor$ counter-clockwise from $a$
		\State Each agent $a$ simultaneously performs:
		\Statex \texttt{SecretTransmit}($i_a,n',b^1_a$)
		\Statex \texttt{SecretTransmit}($i_a,n',b^2_a$)
		\State At round $n'+1$, each agent sends its input around the ring
		\State At round $2n'$ output $q(I)$
	\end{algorithmic}
\end{algorithm}

\begin{theorem}
	\label{theorem:ks-limit}
	In a ring, Algorithm~\ref{alg_ks} is an equilibrium
	when no cheating agent pretends to be more than $n$ agents.
\end{theorem}

\begin{proof}
	Assume by contradiction that a cheating agent pretending to be $d \leq n$ agents
	has an incentive to deviate.
	W.l.o.g., the duplicated agents are $a_1,\dots,a_d$
	(recall the indices $1,\dots,n'$ are not known to the agents).
	
	Let $n'$ be the size of the ring including the duplicated agents,
	i.e., $n' = n+d-1$.
	The clockwise neighbor of $a_{n'}$ is $a_1$,
	denoted $b^1_{a_{n'}}$. 
	Denote $a_c=b^2_{a_{n'}}$ the agent at distance $\lfloor \frac{n'}{2} \rfloor$ counter-clockwise from $a_{n'}$,
	and note that $c \geq d$.
	
	When $a_{n'}$ calls \texttt{Secret-Transmit} to $a_1$,
	$a_{n'}$ holds $R_{n'}$ of that transmission until round $n'-1$.
	When $a_{n'}$ calls \texttt{Secret-Transmit} to $a_c$,
	$a_{c+1}$ holds $X_{n'}$ of that transmission until round $n'-1$.
	By our assumption, the cheating agent duplicated into $a_1,\dots,a_d$.
	Since $d < c+1$,
	the cheater receives at most one piece ($X_{n'}$ or $R_{n'}$)
	of each of $a_{n'}$'s transmissions before round $n'$.
	So, there is at least one input that the cheater does not learn before
	round $n'$.
	According to the Full Knowledge property (Definition~\ref{full_knowledge}),
	for the cheater at round $n'-1$ any output is equally possible,
	so its expected utility for any value it sends is the same,
	thus it has no incentive to cheat regarding the values it sends in round $n'-1$.
	
	\begin{figure}[H]
		\centering
		\begin{minipage}{0.38\textwidth}
			\centering
			\includegraphics[width=0.9\textwidth]{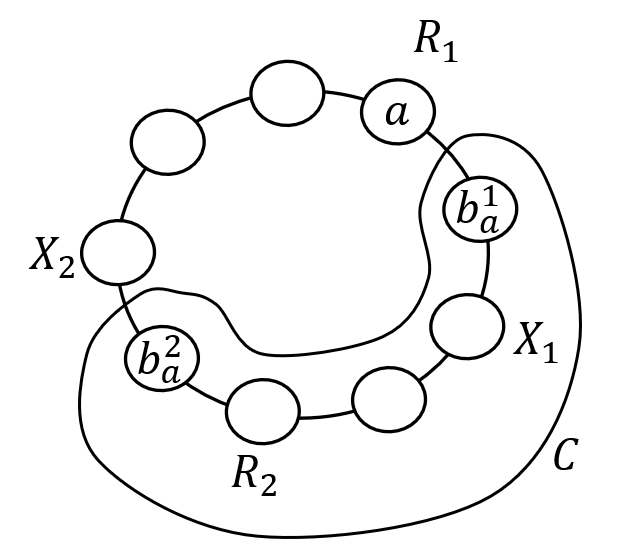} 
			\caption{at round $n'-1$: $i_a=R_1 \oplus X_1=R_2 \oplus X_2$;
				group $C$ does not know the input of $a$}
		\end{minipage}\hfill
		\begin{minipage}{0.45\textwidth}
			\centering
			\includegraphics[width=0.9\textwidth]{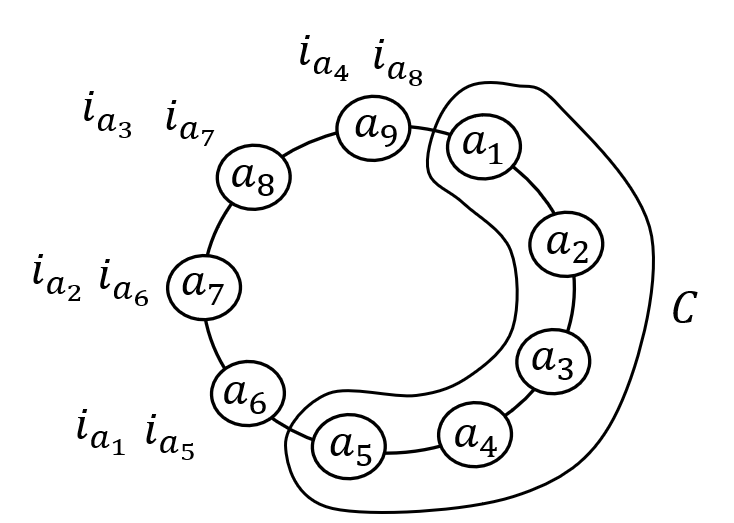}
			\caption{the input that agents $a_6,\dots,a_9$ learn at round $n'$;
				each input of an agent in $C$ is known by an agent not in $C$}
		\end{minipage}
	\end{figure}

	Let $a_j \in \{a_1,\dots,a_d\}$ be an arbitrary duplicated agent.
	In round $n'$, $i_{a_j}$ is known by its clockwise neighbor $b^1_{a_j}$
	and by $b^2_{a_j}$, the agent at distance $\lfloor \frac{n'}{2} \rfloor$
	counter-clockwise from $a_j$.
	Since the number of counter-clockwise consecutive agents in
	$\{b^1_{a_j},a_j,\dots,b^2_{a_j}\}$	is greater than
	$\lceil \frac{n'}{2} \rceil \geq n$,
	at least one of $b^1_{a_j}, b^2_{a_j}$ is not a duplicated agent.
	Thus, at round $n'$, the input of each agent in $\{a_1,\dots,a_d\}$
	is already known by at least one agent $\notin \{a_1,\dots,a_d\}$.
	
	At round $n'-1$ the cheater does not know the input value of at least one other agent,
	so by the Full Knowledge property it has no incentive to deviate.
	At round $n'$ for each duplicated agent,
	its input is already known by a non-duplicated agent,
	which disables the cheater from lying about its input
	from round $n'$ and on.
	
	Thus, the cheating agent has no incentive to deviate, contradicting our assumption.
\end{proof}

\subsection{Coloring in General Graphs}
\label{section:alg:coloring}

Here, agents are given exact a-priori knowledge of $n$.

We present Algorithm~\ref{alg:SemiTotalOrder}, the Witness Algorithm,
that uses two subroutines to obtain a coloring: \emph{Draw} and \emph{Prompt}.

The Witness algorithm includes three colorings: The agents already have an $n$-coloring set by their $id$s. But agents can cheat by lying about their $id$s. So each agent selects a new unique number from $1$ to $n$, in a way that prevents it from cheating about the number. That is another $n$-coloring of the agents, but this coloring does not ensure that every agent $a$ with no neighbors that output $p_a$ will output $p_a$. Thus, an agent that can still output $p_a$ legally, i.e., none of its neighbors are numbered $p_a$, will output $p_a$ \emph{instead} of its unique number, which may also cause an illegal coloring in case two neighbors with the same preferred color behave this way. 
Instead, this second $n$-coloring is used as the order by which agents \emph{choose} their final color, resulting in a legal coloring in which agents have no incentive to cheat.

The algorithm begins by running Wake-Up \cite{Afek:2014:DCB:2611462.2611481} to learn the graph topology
and the $id$s of all agents.
Then, in order of $id$s, each agent $a$ draws a random number $S(a)$
with a neighboring "witness" agent $w(a)$ as specified in Algorithm~\ref{alg:Draw}, and sends it to all of its neighbors.
The number is drawn in the range $1,\dots,n$ and since agents draw numbers one by one, each agent knows the set of numbers already taken by its neighbors and uses it as the input $T$ for its call to \emph{Draw}, ensuring no two neighbors draw the same number and resulting in a coloring of $G$.
However, this coloring is not enough since any equilibrium must ensure an agent that can output its desired color \emph{will} output it.
In \emph{Draw}, 
the agent cannot cheat in the random number generation process since it exchanges a random number simultaneously with its respective witness, 
and $w(a)$ is marked as a witness for future verification, ensuring $a$ will not lie about its drawn number.
When done, each agent simultaneously verifies the numbers published by its neighbors using Algorithm~\ref{alg:prompt},
which enables it to receive each value through two disjoint paths: directly from the neighbor,
and via the shortest simple path to the neighbor's witness that does not include the neighbor.
Then each agent, in order of the values drawn randomly, 
picks its preferred color if available, or the minimal available color otherwise,
and sends its color to all of its neighbors.

\emph{Draw} (Algorithm~\ref{alg:Draw}) is an equilibrium in which agent $a$
randomizes a number different from those of its neighbors and commits to it, and is depicted in Figure~\ref{figure:draw}. For an arbitrary agent $a$ and a subset $T$ of numbers from $\{1 ,\dots, n\}$, \emph{Draw} begins with $a$ sending its neighbor with minimum $id$ the string \texttt{witness}. The following round, $a$ and the agent that received the \texttt{witness} string exchange a random number from $1$ to $|X|$ at the same time, where $X = \{1 \dots n\} \setminus T$. The sum of these two randomized numbers $\mod |X|$, denoted $q$, is calculated at $a$ and at its witness, and then $a$ notifies all its neighbors that its result from the \emph{Draw} process is the $q$'th largest number in $X$, ensuring that when these neighbors use \emph{Draw}, they do not include that number in their respective $X$. \emph{Draw} takes a constant number of rounds to complete, denoted $|Draw|$ rounds.
\emph{Prompt} (Algorithm~\ref{alg:prompt}) is a query that ensures $a$ receives the correct drawn number from a neighbor. 

The complexity of the Witness Algorithm is $O(|E|\cdot n)$.

\begin{algorithm}[H]
	\caption{Witness Algorithm (for agent $a$)}
	\label{alg:SemiTotalOrder}
	\begin{algorithmic}[1]
		
		\State Run Wake-Up \Comment{After which all agents know graph topology}
		\State set $T := \varnothing$
		\Comment{$T$ is the set of values already taken by $a$'s neighbors ($N(a)$)}
		\For{$i=1,...,n$}
		\If{$id_a = i$'th largest $id$ in $V$}
		\Comment{Draw random numbers in order of $id$}
		\State $Draw(T)$
		\Else
		\State \textbf{wait} $|Draw|$ rounds
		\Comment {Wait for $Draw$, takes a constant number of rounds}	
		\If{received $S(v)$ from $v \in N(a)$}
		\Comment{$S(v)$ is the value of $v$ from $Draw$}
		\State $T = T \cup \{S(v)\}$
		\Comment {Add $S(v)$ to set of taken values}
		\EndIf
		\EndIf
		\EndFor
		
		\For{$u \in N(a)$ \textbf{simultaneously}}
		\State $Prompt(u)$  \Comment{Validate the value received in line $8$ using Algorithm~\ref{alg:prompt}}
		\EndFor
		\State \textbf{wait} until all prompts are completed in the entire graph \Comment {At most $n$ rounds}
		\For{round $t = 1,...,n $}:
		\If { $S(a) = t$ }  \Comment {Wait for your turn, decreed by your $S$ value}
		\If { $\forall v \in N(a) : o_v \neq p_a$}  $o_a := p_a$
		
		\Else { $o_a := $ minimum color unused by any $v \in N(a)$}
		\EndIf 
		\State \textbf{send} $o_a$ to $N(a)$
		\EndIf
		\EndFor
		
	\end{algorithmic}
\end{algorithm}

\begin{algorithm}[H]
	\caption{Draw($T$) Subroutine (for agent $a$ and input $T \subseteq \{1 ,\dots, n\}$)}
	\label{alg:Draw}
	\begin{algorithmic}[1]
		\Statex Denote $X=\{1,...,n\}\setminus{T}$
		\Comment {X is the set of numbers not drawn by neighbors}
		\State $w(a) :=$ node $b$ s.t., $id_b$ is minimal in N(a)
		\Comment {$N(a)$ is the set of neighbors of $a$}
		\Statex \textbf{send} the string \texttt{witness} to $w(a)$
		\Comment {choose neighbor with minimal $id$ as witness} 	
		\State $r(a) := random\{1,...,|X|\}$ drawn by $a$
		\Statex $r(w(a)) := random\{1,...,|X|\}$ drawn by $w(a)$ 
		\Statex \textbf{send} $r(a)$ to $w(a)$
		\Statex \textbf{receive} $r(w(a))$ from $w(a)$
		\Comment{$a$ and $w(a)$ exchange $r(a), r(w(a))$ at the same time}
		\State Let $q := r(a)+r(w(a))\ mod\ |X|$.
		\Statex Set $S(a) := q$'th largest number in $X$
		\Statex \textbf{send} $S(a)$ to all $u \in N(a)$ \Comment {Calculate $S(a)$ and publish to neighbors} 
	\end{algorithmic}
\end{algorithm}

\begin{figure}[H]
	\caption{Communication in a \emph{Draw} algorithm commenced by $a_1$}	
	\label {figure:draw}
	\captionsetup{justification=centering}
	\includegraphics[scale=0.5]{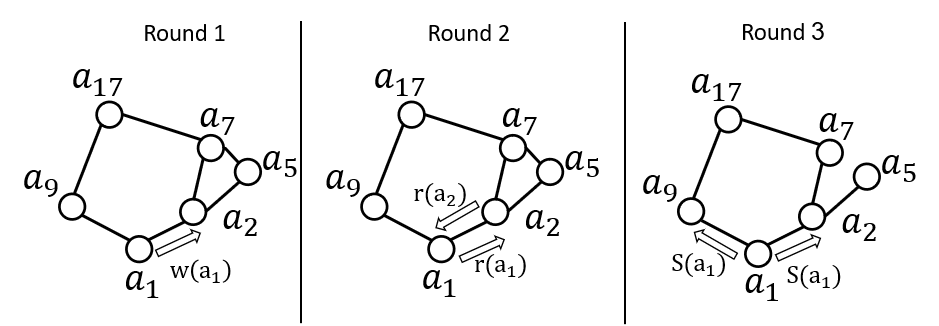}
	\centering
\end{figure}

\pagebreak

\begin{algorithm}[H]
	\caption{Prompt$(b)$ Subroutine (for agent $a$)}
	\label{alg:prompt}
	\begin{algorithmic}[1]
		\Statex upon receiving a $prompt(b)$ message from $b \in N(a)$:
		\State $p := \text{ shortest simple path  } a \rightarrow w(a) \rightarrow b$
		\Comment{$w(a)$ is set by a preceding call to \texttt{Draw}}
		\Statex \textbf{send} $S(a),b $ via $p$ 
		\Comment{If $v \neq w(a)$ is asked to relay $S(a)$, $v$ fails the algorithm}
		\Statex \textbf{send} $S(a)$ to $b$ via $e = (a,u)$
		\Comment {$b$ validates that both messages received are consistent}
	\end{algorithmic}
\end{algorithm}

\begin{figure}[H]
	\caption{Communication in a \emph{Prompt($a_1$)} commenced by $a_9$}	
	\label {figure:prompt}
	\captionsetup{justification=centering}
	\includegraphics[scale=0.4]{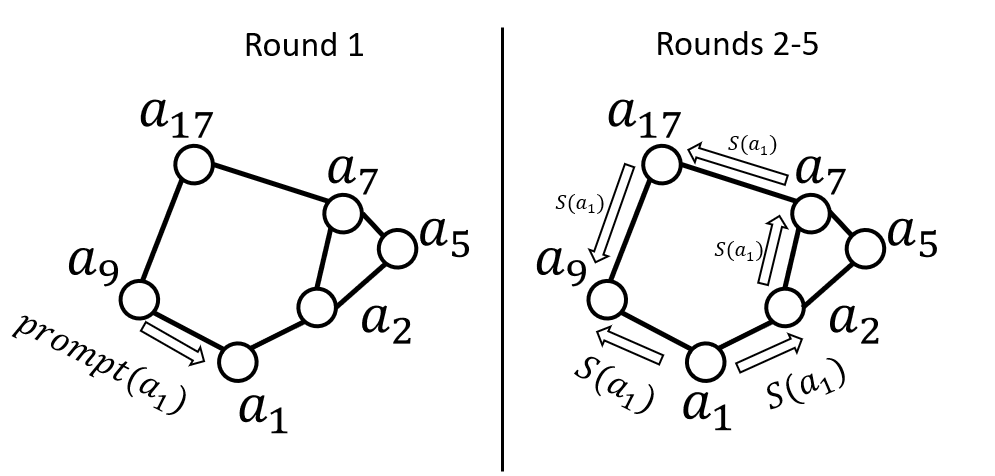}
	\centering
\end{figure}

The resulting message complexity is:
\begin{itemize}
	\item \emph{Wake Up} is $O(|E| \cdot n)$
	\item Drawing a random number is called $n$ times and uses a total of $O(|E|)$ messages in total to publish these values to neighbors
	\item Verifying the value of a neighbor uses $O(diameter)$ messages and is called $O(|E|)$ times, for a total of $O(|E| \cdot diameter)$
	\item Sending the output color takes an additional $O(|E|)$ messages
\end{itemize}

Since the diameter of the network is at most $n$ (line graph), we have a total of $O(|E| \cdot n)$.

\begin{theorem}
	The Witness Algorithm (Algorithm~\ref{alg:SemiTotalOrder}) is an equilibrium for Coloring when agents a-priori know $n$.
\end{theorem}

\begin{proof}
	Let $a$ be an arbitrary agent. Assume in contradiction that at some round $r$ there is a
	possible cheating step $s$ such that $s \neq s_r$
	and $ \mathbb{E}_{s,r}[u_a] > \mathbb{E}_{s_r,r}[u_a]$.
	Consider the possible deviations for $a$ in every phase of Algorithm~\ref{alg:SemiTotalOrder}:
	\begin{itemize}
		\item Wake-Up:
		The order by which agents initiate Algorithm~\ref{alg:Draw} has no effect on the order by which they will later set their colors.
		Hence, $a$ has no incentive to publish a false $id$ in the Wake-Up building block.
		\item \emph{Draw} is an equilibrium: The topology and $id$s are known to all agents from Wake-Up, so all agents know who should be the witness of each agent (thus no agent can use \emph{Draw} with the wrong witness). Agent $a$ and its witness then exchange a random number at the same time, so $a$ cannot affect the draw: Its expected utility is the same regardless of the number it sends in the exchange.
		\item Publishing a false $S$ value will be caught by the verification in step 10 of Algorithm~\ref{alg:SemiTotalOrder}.
		\item Sending a color message not in order will be immediately recognized by the neighbors, since $S$ values were verified.
		\item Agent $a$ might output a different color than the color dictated by Algorithm~\ref{alg:SemiTotalOrder}. But if the preferred color is available,
		then outputting it is the only rational behavior.
		Otherwise, the utility for the agent is already $0$ in any case. \qedhere
	\end{itemize}
\end{proof}

\section{Impossibility With No Knowledge}
\label{section_imp}

Here we prove that the common assumption that $n$ is known is the key
to the possibility of equilibrium for the Knowledge Sharing problem. 
When agents know that $n \sim \sigma$ where $\sigma$ has infinite support, there is \emph{no} equilibrium for the Knowledge Sharing problem. In other words, agents \emph{must} have some a-priori information that $n$ is part of a finite set.

In this section we label agents in the graph as $a_1,...,a_n$
in an arbitrary order in any topology.
These labels are not known to the agents themselves.


\subsection{Impossibility of Knowledge Sharing}
\label{know_share}


\begin{theorem}
	\label{theorem:KSgeneral}
	There is no protocol for Knowledge Sharing
	that is an equilibrium in a $2$-connected graph when agents' prior on $n$ has infinite support.
	Formally, for any protocol $\Lambda$ for Knowledge Sharing there exists a graph $G$ where agents know $n \sim \sigma$ when $\sigma$ has infinite support such that there exists an agent $a$ and a strategy $D$ for $a$ such that:
	\[
		\ev_{D}[u_a] > \ev_{\Lambda}[u_a]
	\]
\end{theorem}

\begin{proof}
	Assume by contradiction that  $\Lambda$ is a Knowledge Sharing algorithm
	that is an equilibrium in any graph of agents with a prior with infinite support on $n$. 
	Let $D$, $E$ be two $2$-connected graphs of rational agents.
	Consider the execution of $\Lambda$ on graph $H$ created by $D,E$, and adding two nodes $a_1,a_2$ and connecting
	these nodes to $1$ or more arbitrary nodes in both $D$ and $E$ (see Figure~\ref{figure:general_graph_clear}).

	\begin{figure}[H]
		\centering
		\begin{minipage}{0.45\textwidth}
			\centering
			\includegraphics[width=0.9\textwidth]{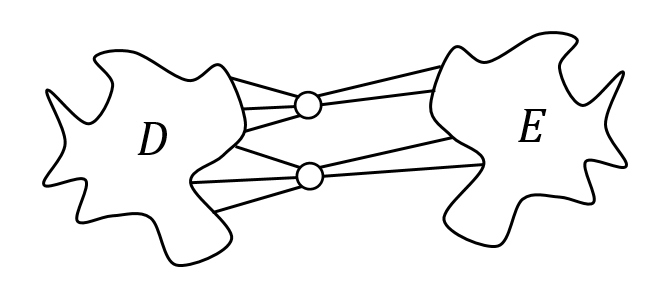} 
			\caption{Graph $H$ created by two arbitrary sub-graphs $D$,$E$}
			\label{figure:general_graph_clear}
		\end{minipage}\hfill
		\begin{minipage}{0.45\textwidth}
			\centering
			\includegraphics[width=0.9\textwidth]{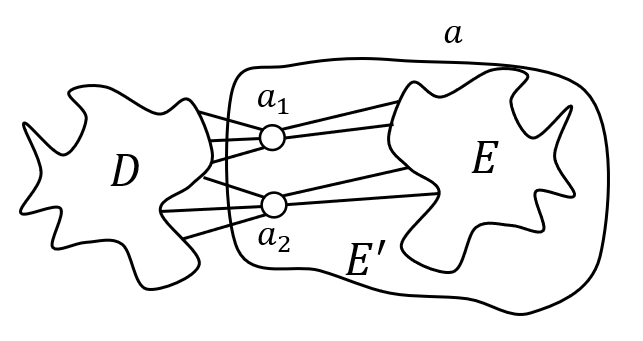}
			\caption{Example of agent $a$ pretending to be $E'=E \cup \{a_1,a_2\}$}
			\label{figure:general_graph_duplicated}
		\end{minipage}
	\end{figure}
	 
	Recall that the vector of agents' inputs is denoted by $I = i_1 , i _2 , \cdots , i_n$,
	and $n=|H|=|D|+|E|+2$.
	Let $t_D$ be the first round after which $q(I)$ can be calculated
	from the collective information that all agents in $D$ have (regardless of the complexity of the computation),
	and similarly $t_E$ the first round after which $q(I)$ can be calculated in $E$.
	Consider the following three cases:

	\begin{enumerate}
		\item $\bm{t_E<t_D}$: $q(I)$ cannot yet be calculated in $D$ at round $t_E$.
		Let $E'=E \cup \{a_1,a_2\}$.
		Since $E \subset E'$,
		the collective information in $E'$ at round $t_E$ is enough to calculate $q(I)$.
		Since $n$ is not known, an agent $a$ could emulate the behavior of $E'$,
		making the agents believe the algorithm runs on $H$ rather than $D$.
		In this case, this cheating agent knows at round $t_E$ the value of $q(I)$ in this execution,
		but the collective information of agents in $D$ is not enough to calculate $q(I)$,
		which means the output of agents in $D$ still depends on messages from $E'$, the cheater.
		Thus, if $a$ learns that the output $q(I) \neq p_a$,
		it can simulate all possible runs of the algorithm in a state-tree, and select a course of action that has at least some probability of leading to an outcome $q(I) = p_a$. Such a message surely exists because otherwise, $D$ would have also known the value of $q(I)$. In other words, $a$ finds a set of messages that \emph{might} cause the agents in $D$ to decide a value $x \neq q(I)$.
		In the case where $p_a=x$, agent $a$ increases its expected utility by sending
		a set of messages different than that decreed by the protocol.
		Thus, agent $a$ has an incentive to deviate, contradicting distributed equilibrium.
		\label{case:t_great_full_graph}
		
		\item $\bm{t_D=t_E}$: both $E$ and $D$ have enough collective information to calculate $q(I)$
		at the same round.
		The collective information in $E$ at round $t_E$ already exists in $E'$ at round $t_E-1$.

		Since $t_D=t_E$, the collective information in $D$ is not enough to calculate $q(I)$ in round $t_E-1$.
		Thus, similarly to Case~\ref{case:t_great_full_graph}, $a$ can emulate $E'$ and has an incentive to deviate.
		
		\item $\bm{t_E>t_D}$: Symmetric to Case~\ref{case:t_great_full_graph}.
	\end{enumerate}
	
	Thus, $\Lambda$ is not an equilibrium for the Knowledge Sharing problem.
\end{proof} 

\section{How Much Knowledge Is Necessary?}
\label{section_approx}
In this chapter, agents know that $n \sim \sigma$ where $\sigma$ is a prior on $n$ with finite support,
	i.e., $\sigma \subseteq [\alpha,\beta]$ for some $\alpha,\beta \in \mathbb{N}$.
We examine several distributed computing algorithms and find bounds on the interval $[\alpha,\beta]$ for which we have equilibria. We use two test-cases for possible priors, namely, a uniform prior and a geometric prior. 
We begin by showing bounds on the interval  $[\alpha,\beta]$ that for any $\sigma \subseteq [\alpha,\beta]$, Algorithm~\ref{alg_ks} is an equilibrium for Knowledge Sharing in a ring network. Then, we use a subset of these results to show a new algorithm for Coloring a ring network, and show that if Algorithm~\ref{alg_ks} is an equilibrium when $n \sim \sigma$, then the Coloring algorithm provided is an equilibrium when $n \sim \sigma$. 

Finally, we look into bounds on $\alpha$,$\beta$ for Leader Election, Partition and Orientation.

Notice that of all the problems we examine, Partition and Orientation are an exception: We show that they have equilibria without \emph{any} knowledge of $n$, i.e., even if the support agents have for $\sigma$ is infinite;
however, the former is constrained to even-sized rings,
and the latter is a trivial problem in distributed computing
(radius $1$ in the \textit{LOCAL} model \cite{Linial:1987:DGA:1382440.1382990}).

Consider an agent $a$ at the start of a protocol:
If $a$ pretends to be a group of $d$ agents, it is caught when $d + n - 1 > \beta$,
since when agents count the number of agents they identify that a deviation occurred.
The finite support for $\sigma$ creates a situation where \emph{any} duplication now involves some risk
since the actual value of $n$ is not known to the cheater
(similar to \cite{gt_idspace}).

Let $\Gamma$ denote the strategy of following the protocol, and for any agent $a$ let $\ev_{\Gamma}[u_a]$ be the expected utility for $a$ by following the protocol without any deviation.  
An arbitrary cheating agent $a$ simulates executions of the algorithm for every possible duplication,
and evaluates its expected utility compared to $\ev_{\Gamma}[u_a]$. 
Denote $D$ a duplication scheme in which an agent pretends to be $d$ agents.

Let $P_D = P[d + n - 1 > \beta]$ be the probability,
from agent $a$'s perspective, 
that the overall size of the network exceeds $\beta$. 
If for agent $a$ there exists a duplication scheme $D$ at round $0$ such that
$ \mathbb{E}_{D,0}[u_a] \cdot (1-P_D)  > \mathbb{E}_{s(0),0}[u_a]$,
then agent $a$ has an incentive to deviate and duplicate itself \footnote{Agents are assumed to duplicate at round $0$ because any other duplication can be detected by adding a \emph{Wake-Up} routine which starts by mapping the topology. We consider round 0 as deciding on $d$ before the \emph{Wake-Up} is completed.}

\subsection{Knowledge Sharing in a Ring}
\label{class_KS}	

In ring networks, we have shown Algorithm~\ref{alg_ks} is an equilibrium if no agent duplicates $d > n$. In other words, the only way an agent $a$ can increase its expected utility is by duplicating into at least $2$ more agents than the size of the remainder of the ring, i.e., $n'$ is made of $n-1$ honest agents and at least $n+1$ duplicated cheater agents. It is easy to see that if a cheating agent had duplicated $d > n$ agents and was not detected (i.e., $d + n - 1 \leq \beta$), its utility is $1$ since it controls the output. Notice that Algorithm~\ref{alg_ks} begins with \emph{Wake-Up}, so a duplicating agent must decide on its duplication scheme before any communication other than \emph{Wake-Up} occurs. 

We consider a specific duplication scheme for the cheating agent: Duplicating $d = \lfloor \frac{\beta}{2} \rfloor + 1$. This cheating strategy maximizes a cheater's expected utility by obtaining the optimal balance between $P_D$ and successfully duplicating $d > n$ (recall that $n$ is a random variable from the cheater's perspective).
Notice that duplicating $d> n$ is beneficial \emph{only} when $\beta \geq 2\alpha$. If $\beta < 2\alpha$, duplicating $d = \lfloor \frac{\beta}{2}  \rfloor + 1$ would either be detected and lead to $0$ utility, or in case $n = \alpha$, would result in $d = n$ and not increase the cheater's utility. Any duplication other than $d = \lfloor \frac{\beta}{2}  \rfloor + 1$  would also result in $d \leq n$ (which does not improve the cheater's expected utility) or in detection.

Let $a$ be an agent in the ring at round $0$ with prior $n \sim \sigma$.
Let $D$ be a cheating strategy in which $a$ duplicates into  $d = \lfloor \frac{\beta}{2} \rfloor +1$ agents,
and if $a$ is not caught, it skews the protocol to output $p_a$.
The expected utility by following strategy D is:

\begin{multline*}
\ev_{D,0}[u_a] = 0 \cdot \prob[\text{detected}] + 1 \cdot \prob[\text{undetected, duplicated } d>n] +
	\frac{1}{k} \cdot \prob[\text{undetected, duplicated } d \leq n] = \\
1 \cdot \prob[n < d \land d +n - 1 \leq \beta] + \frac{1}{k} \prob[ d \leq n \land d +n - 1 \leq \beta ]
= \\ =
1 \cdot \prob[n \leq \min (d-1, \beta - d + 1)]  + \frac{1}{k}
\big( \prob[n \leq \beta - d +1] -\prob[n \leq d - 1] \big)
\end{multline*}

If for all agents this utility is at most $\ev_{\Gamma}[u_a] = \frac{1}{k}$, the expected utility from following the protocol, then Algorithm~\ref{alg_ks} is an equilibrium. Setting the probability mass function F and
$d = \lfloor \frac{\beta}{2} \rfloor +1$ yields the following requirement for Algorithm~\ref{alg_ks} to be an equilibrium:

\begin{equation}
\label{eq:cheater_eu_ks}
\ev_{D,0}[u_a] =
F(\lfloor \frac{\beta}{2} \rfloor) +
\frac{1}{k} \bigg( F(\lceil \frac{\beta}{2} \rceil) - F(\lfloor \frac{\beta}{2} \rfloor) \bigg) 
\leq
\frac{1}{k}
\end{equation}

For each prior and each Knowledge Sharing variant, if $\ev_{D,0}[u_a] > \frac{1}{k}$ then there is an incentive to cheat by duplicating $d = \lfloor \frac{\beta}{2} \rfloor + 1$, and then using an optimal cheating strategy that derives a utility of $1$. Table~\ref{table:ks_bounds} shows the conditions on $[\alpha,\beta]$ for Algorithm~\ref{alg_ks} to be an equilibrium.
Following the table are explanations of each entry.

\begin{table}[ht]
	\centering
	\begin{tabular}{|c|c|c|}
		\hline
		& $\sigma$ uniform &  $\sigma$ geometric   \\ \hline
		$k=2$
		& \shortstack{Equilibrium for any \\ $\alpha,\beta$ (1) }

		& \shortstack{Equilibrium iff \\ $\beta \leq 2\alpha - 1$} (4)
		
		\\ \hline
		\shortstack{Constant \\ $k > 2$}
		& \shortstack{ Equilibrium iff
			\\ $\beta \leq \frac{(2\alpha - 2)(k-1)}{k-2}$ (2)}
		
		&  \shortstack{ Equilibrium iff \\ $\beta \leq 2\alpha - 1$} (4)
		
		\\ \hline
		\shortstack{Equilibrium  \\ for any $k$}
		&  \shortstack{Equilibrium iff \\ $\beta \leq 2\alpha - 1$} (3)
		
		& \shortstack{ Equilibrium iff \\ $\beta \leq 2\alpha - 1$} (4)
		
		\\ \hline
	\end{tabular}
	\caption{Requirements of $\alpha,\beta$ for equilibrium in different Knowledge Sharing settings}
	\label{table:ks_bounds}
\end{table}

\begin{enumerate}
	\item For uniform $\sigma$ and $k=2$, setting the probability mass function $F$ for the uniform distribution into Equation~\ref{eq:cheater_eu_ks} yields:
	\[
		\frac{\lfloor \frac{\beta}{2} \rfloor - \alpha + 1}
		{\beta - \alpha + 1}
		+ \frac{1}{2}
		\bigg(
			\frac{ \lceil \frac{\beta}{2} \rceil - \lfloor \frac{\beta}{2} \rfloor}
				 {\beta - \alpha + 1}
		\bigg)
		\leq
		\frac{1}{2}
	\]
	
	\[
		\frac{\lfloor \frac{\beta}{2} \rfloor - \alpha + 1}
		{\beta - \alpha + 1}
		\leq
		\frac{1}{2}
		\bigg(
			\frac{\beta - \alpha +1 -
						\lceil \frac{\beta}{2} \rceil 
						+
						\lfloor \frac{\beta}{2} \rfloor}
					{\beta - \alpha + 1}
		\bigg)
	\]
	
	\[
		\lfloor \frac{\beta}{2} \rfloor - \alpha + 1
		\leq
		\frac{1}{2}
		\bigg(
			2 \cdot \lceil \frac{\beta}{2} \rceil - \alpha + 1
		\bigg)
	\]
	
	\[
		2 \cdot \lfloor \frac{\beta}{2} \rfloor
		- 2 \alpha + 2
		\leq
		2 \cdot \lfloor \frac{\beta}{2} \rfloor
		-\alpha + 1
	\]
	\[
		\alpha \geq 1
	\]

	This inequality holds for any $\alpha,\beta$. Thus when $k=2$ and $\sigma$ is a uniform prior, Algorithm~\ref{alg_ks} is an equilibrium for any $\alpha,\beta$. 
	
	\item When $\sigma$ is uniform and $k > 2$, setting the probability mass function $F$ for the uniform distribution into Equation~\ref{eq:cheater_eu_ks}:
	\[
	\frac{\lfloor \frac{\beta}{2} \rfloor - \alpha + 1}
	{\beta - \alpha + 1}
	+ \frac{1}{k}
	\bigg(
	\frac{ \lceil \frac{\beta}{2} \rceil - \lfloor \frac{\beta}{2} \rfloor}
	{\beta - \alpha + 1}
	\bigg)
	\leq
	\frac{1}{k}
	\]
	
	\[
	\frac{\lfloor \frac{\beta}{2} \rfloor - \alpha + 1}
	{\beta - \alpha + 1}
	\leq
	\frac{1}{k}
	\bigg(
	\frac{\beta - \alpha + 1
		+ \lfloor \frac{\beta}{2} \rfloor -  \lceil \frac{\beta}{2} \rceil}
	{\beta-\alpha + 1}
	\bigg)
	\]
	
	\[
	(\lfloor \frac{\beta}{2} \rfloor - \alpha + 1)
	\leq
	\frac{1}{k}
	\big(
	\beta - \alpha + 1
	+ \lfloor \frac{\beta}{2} \rfloor -  \lceil \frac{\beta}{2} \rceil
	\big)
	=
	\frac{1}{k}
	\big(
	2 \cdot \lfloor \frac{\beta}{2} \rfloor - \alpha + 1
	\big)
	\]
	
	\begin{equation}
	\label{eq:min_X_uniform}
	k
	\leq
	\frac{2 \cdot \lfloor \frac{\beta}{2} \rfloor - \alpha + 1}
	{\lfloor \frac{\beta}{2} \rfloor - \alpha + 1}
	= 
	1 + \frac{\lfloor \frac{\beta}{2} \rfloor}
	{\lfloor \frac{\beta}{2} \rfloor - \alpha + 1}
	\end{equation}
	
	We can then consider two cases for $\beta$:
	
	\begin{itemize}
		\item $\beta$ is even:
		\[
		k \leq 1 + \frac{\beta}{\beta - 2\alpha + 2}
		\]
		\[
			\beta \leq \frac{(2\alpha - 2)(k-1)}{k-2}
		\]
		\item $\beta$ is odd, then using the same derivation:
		\[
			k \leq \frac{\beta - 1}{\beta - 1 - 2\alpha + 2}
		\]
		\[
			\beta \leq \frac{2\alpha(k-1) - k}{k-2}
		\]
		
	\end{itemize}
	
	Notice that for any $k > 2$, the even upper bound is lower, so  Algorithm~\ref{alg_ks} is an equilibrium if and only if: 
		\[
			\beta \leq \frac{(2\alpha - 2)(k-1)}{k-2}
		\]
	
	\item We now want to know the bounds on $[\alpha,\beta]$ for which we have a general Knowledge Sharing protocol that is an equilibrium, namely, a Knowledge Sharing protocol that is an equilibrium for any $k$. Notice that we assume agents always know $k$ beforehand, but we require that the same equilibrium should apply to any $k$. We can see that the demand for an equilibrium for constant $k$, namely:
	\[
		\beta \leq \frac{(2\alpha - 2)(k-1)}{k-2}
	\]
	
	Actually shows that as a function of $k$, the upper bound for an equilibrium converges to $\beta \leq 2\alpha - 1$. In other words, for any uniform $\sigma$ with $\beta \geq 2\alpha$, there exists a $k$ large enough so that the protocol is \emph{not} an equilibrium, i.e., there exists an agent with an incentive to cheat. As stated above, if $\beta \leq 2\alpha -1$, Algorithm~\ref{alg_ks} is an equilibrium, because no agent can duplicate $d>n$ without surely being detected.

	
	\item When $\sigma$ is geometric, consider Equation~\ref{eq:cheater_eu_ks}. Notice that if the inequality holds, Algorithm~\ref{alg_ks} is an equilibrium. If $\beta$ is even, the left hand side of the inequality is smaller. Assume that $\beta$ is even, we get that Equation~\ref{eq:cheater_eu_ks} yields that Algorithm~\ref{alg_ks} is an equilibrium when :
	\[
		F(\frac{\beta}{2}) \leq \frac{1}{k}
	\]
	
	By definition of $F$ for a geometric prior, $\prob[n=\alpha] > \frac{1}{2}$. In other words, for any $\beta \geq 2\alpha$ and any $k \geq 2$, $F(\frac{\beta}{2}) > \frac{1}{2} \geq \frac{1}{k}$. So a cheating agent $a$ has an incentive to deviate for any $k$ and any $\beta \geq 2\alpha$. The same holds for the case where $\beta$ is odd.

\end{enumerate}


\subsection{Coloring}
\label{class_color}

In this section we show that for any ring and any prior on $n$ with finite support $\sigma \subseteq [\alpha,\beta]$ in which Algorithm~\ref{alg_ks} is an equilibrium for $2$-Knowledge Sharing, there is an equilibrium for Coloring.
We describe Algorithm~\ref{alg:coloring_ring}, which uses a shared coin toss subroutine.
Using Algorithm~\ref{alg_ks}, we set each agent's input be a random bit
and define the output function $q(\vec{I}) = \sum_{i \in I} i \mod 2$.
This shared coin toss subroutine is an equilibrium as a subroutine if agents' utilities are
either a direct result of the output of the coin flip, or completely independent of it. 
	
	\begin{algorithm}[H]
		\caption{Coloring in a Ring}
		\label{alg:coloring_ring}
		\begin{algorithmic}[1]
			\State Every agent sends its $id$ to both of its neighbors
			\State Forward every $id$ to the other direction until all agents have $\vec{id} = id_0,..., id_{n'-1}$ and agents know $n'$, which may include duplicates.
			$\vec{id}$ is sent across the ring for validation. If an agent finds an inconsistency it outputs $\bot$.
			Using the $id$s agents have agreed on a global orientation around the ring: The orientation of the ring is set to be from the agent with the highest $id$ to its neighbor with the lowest $id$, either clockwise or counter-clockwise.
			\State Every agent sends its preference $p_a$ to both its neighbors.
			\State Forward every preference until all agents have $\vec{p} = p_0,...,p_{n'-1}$.
			\State $\vec{p}$ is sent across the ring again. If an agent finds an inconsistency it outputs $\bot$.
			\If {$n'$ is odd and all agents have the same preferred color:}
			\begin{enumerate}
				\item[a.] Every agent $a$ sends a random integer $w_a \in [0,n'-1]$ in the direction of the orientation (clockwise or counter-clockwise).
				\item[b.] Agents forward the $w_a$s they receive until they know all $\vec{w} = w_0,..., w_{n'}$.
				\item[c.] $\vec{w}$ is sent across the ring. If an agent finds an inconsistency it outputs $\bot$ and aborts.
				\item[d.] Every agent knows $s$, the "sink", is the $w=(\sum_a w_a) \mod n'$ th largest $id$.
			\end{enumerate}
			\EndIf
			\State Agents flip a shared coin, denoting its result $b$, using $2$-Knowledge Sharing.
			\State Let $C$ be a group of consecutive agents with the same preferred color, excluding $s$ if a sink exists,
			and denote $c_1,..., c_{|C|}$ the agents in $C$ by order of the orientation.
			If $C$ is the entire ring (i.e., the ring is even and all agents prefer the same color),
			let $c_1$ be the agent with maximum $id$. (This step is merely for notation purposes, agents do not actively exchange messages as they can deduce these groups from the orientation and the preference vector $\vec{p}$).
			\State For every group $C$,
			each agent $c_i$, s.t., $i \mod 2 = b$ (except $s$ in case there is a sink) where $b$ is the result of the shared coin from Step 7,
			sets its output to its preferred color $o_{c_i} = p_{c_i}$ and sends it to its neighbors.
			An agent with no neighbors that prefer its color, i.e., agent $c_1$ for any group $C$ such that $|C|=1$,
			also outputs its preferred color and sends it to its neighbors in this round.
			\State All agents which did not output a color yet set their output to the minimum available color and send it to their neighbors. The only agents that do not output at this step are:
			\begin{itemize}
				\item $s$ if a sink exists, i.e., the ring is odd and all agents prefer the same color $p_a$.
				\item The last agent (according to the orientation) in every consecutive group of agents with the same preference, i.e., for every group $C$, the agent $C_{|C|}$. 
			\end{itemize} 
			\State For every group $C$,
			each agent $c_{|C|}$ that did not set its output yet,
			sets its output to the minimum available color and sends it to its neighbors.
			\State If a sink $s$ exists, it sets its output to the minimal available color.
		\end{algorithmic}
	\end{algorithm}

	\begin{theorem}
		\label{thm:ring_coloring_legal}
		Algorithm~\ref{alg:coloring_ring} always ends in a legal coloring of the ring
	\end{theorem}
	
	\begin{proof}
		\begin{enumerate}
			\item Every agent outputs a color. All agents Know the orientation and the order of agents around the ring according to it, and the vector of preferences. Out of these they can deduce their preference group $C$ and their position in it. All agents with parity $b$ output a color at step $9$. The rest of the agents output at step $10-11$. The only exception is the sink, if there is one, which outputs at step $12$. Any agent either has an even parity, an odd parity or is the sink. Thus any agent has chosen and output a color.
			\item For any neighbors $a_1, a_2$, the round in which $a_1$ chooses its color is different from the round $a_2$ chooses its color. This is immediate: Since they are neighbors, their parity is different, so if neither of them is the sink, one would choose a color at step $9$ and the other at step $10-11$. Separation of steps $10$ and $11$ is merely to prevent a monochromatic edge between two different preference groups when both "edge" agents did not receive their desired color (and outputting the minimal available color may collide). In this case, both agents already have a sure utility of $0$. If one of them is the sink, assume w.l.o.g. it is $a_1$, then it chooses a color at step $12$, and there is only one sink so $a_2$ chooses at step $9$ or steps $10-11$. 
			\item No agent would choose a color already taken. This is immediate from Solution Preference. 
		\end{enumerate}
		
	\end{proof}
	
	\begin{theorem}
		\label{theorem:coloring:bound}
		Let G be a ring network and $\sigma$ a finite-support prior that agents have on $n$. If Algorithm~\ref{alg_ks} is an equilibrium for $2$-Knowledge Sharing in $G$, then Algorithm~\ref{alg:coloring_ring} is an equilibrium for Coloring in G.
	\end{theorem}

	
	\begin{proof}
		Denote $\ev_{\Gamma}[u_a]$ the expected utility for agent $a$ by following Algorithm~\ref{alg:coloring_ring}. Our proof uses $4$ claims:
		\begin{claim}[1]
			For any agent $a$ and any strategy $T$ that does not include duplication, $\ev_{\Gamma}[u_a] \geq \ev_{T}[u_a]$. In other words, no agent has an incentive to deviate without duplication.
		\end{claim}
		\begin{claim}[2]
			For any agent $a$ and any duplicating strategy $D$, let $t$ be the round in which $a$ learns $n$. Then one of the following holds:
			\begin{enumerate}
				\item $a$ had already stopped duplicating at round $t$, i.e., $d < n$.
				\item $a$ had duplicated at least $n$ agents.
			\end{enumerate}
		\end{claim}
		\begin{claim}[3]
			If $d>n$ then the probability that $a$ is detected and the algorithm is aborted is at least $\frac{1}{2}$.
		\end{claim}
		\begin{claim}[4]
			If $d \leq n$, the expected utility for $a$ is at most the same as following the protocol without duplication.
		\end{claim}

		These 4 claims suffice to prove that Algorithm~\ref{alg:coloring_ring} is an equilibrium: If $a$ does not duplicate, it has no incentive to deviate (Claim 1), so any possible cheating strategy involves duplication. Let $d$ be the number of duplications $a$ disguises itself as. According to Claim 2, $a$ decides whether to duplicate into $d > n$ before learning $n$. Then one of the following holds:
		\begin{enumerate}
			\item $d \leq n$. Then by Claim 4 the cheater's expected utility by deviating is at most that of following the protocol without deviation.
			\item $d > n$. By Claim $3$, the probability of detection is at least $\frac{1}{2}$. In other words, the probability the duplication is \emph{not} detected is at most $\frac{1}{2}$. Thus, even if $d>n$ guarantees $a$ a utility of $1$ when it is not detected, the expected utility from duplicating into $d$ agents is at most $\frac{1}{2}$. 
		\end{enumerate}
		
		Thus for any $d$, the expected utility from deviating using duplication is at most $\max \{\frac{1}{2}, \ev_{\Gamma}[u_a]\}$.
		
		On the other hand, consider the expected utility for agent $a$ following Algorithm~\ref{alg:coloring_ring} with preferred color $p_a$:

		\begin{itemize}
			\item Let $\prob[u]$ be the probability given by $a$, according to its prior beliefs, that $n$ is odd and all agents prefer $p_a$.
			\item Let $\prob[A]$ be the probability given by $a$, according to its prior beliefs, that at least one of its neighbors prefer $p_a$.
			\item Let $\prob[s]$ be the probability that $a$ is selected as the sink in a sink-election process. 
		\end{itemize}
		
		Now consider the expected utility for $a$ before it learns $n$, i.e., before the \emph{Wake-Up} is complete:
		If none of its neighbors prefer $p_a$, $a$ will be colored $p_a$, thus its utility in this case will be $1$. If at least one neighbor prefers $p_a$, one of the following must hold:
		\begin{itemize}
			\item Agents are in an odd ring and all agents prefer $p_a$. Then $a$ is colored $p_a$ if it is not the sink \emph{and} it wins in the coin flip.
			\item Otherwise, $a$ is colored $p_a$ if it wins the coin flip.
		\end{itemize}
		
		Recall that $\ev_{\Gamma}[u_a]$ denotes the expected utility for $a$ by following Algorithm~\ref{alg:coloring_ring}. Then we have: 
		
		\begin{multline}
		\label{eq:ring_color_rules}
		\ev_{\Gamma}[u_a]
		= 
		(1-\prob[A]) \cdot 1 + \prob[A]
		\bigg(
		\prob[u|A](1-\prob[s])\frac{1}{2}
		+
		(1-\prob[u|A]) \frac{1}{2})
		\bigg)
		= \\ =
		1 - \prob[A] + \frac{\prob[A]}{2}
		\bigg(
		\prob[u|A] - \prob[u|A]\prob[s] + 1 - \prob[u|A]
		\bigg)	
		= \\ =
		1-\prob[A] - \frac{\prob[A] \cdot \prob[u|A]\prob[s]}{2} + \frac{\prob[A]}{2}
		= \\ =
		1 - \frac{\prob[A]}{2} -  \frac{\prob[A] \cdot \prob[u|A]\prob[s]}{2}
		\end{multline}
		
		Recall that we assume the agents' prior on the preferences of other agents is uniform. Denote $c$ the set of possible color preferences. so $\prob[A] = 1 - (\frac{c-1}{c})^2$. 
		Since agents know there are at least $3$ colors (otherwise, there would be no coloring for an odd ring), $\prob[A] = 1 - (\frac{c-1}{c})^2 \leq \frac{5}{9}$. Furthermore, in a sink selection process where agents follow the protocol the probability that $a$ is selected as the sink is $\prob[s] = \frac{1}{n}$. 
		
		Thus we can bound the term in Equation~\ref{eq:ring_color_rules} from below:
		\[
		\ev_{\Gamma}[u_a] = 
		1 - \frac{\prob[A]}{2} -  \frac{\prob[A] \cdot \prob[u|A]\prob[s]}{2}
		\geq 1 - \frac{\frac{5}{9}}{2} - \frac{\frac{5}{9}}{2} \frac{1}{n} \prob[u|A]
		=
		\frac{13}{18} - \frac{5}{18n} \prob[u|A]
		\]
		
		Finally, $\prob[u|A] \leq 1$ and $n \geq 2$ (otherwise, $a$ is a lone agent). Thus we complete the bound:
		\[
		\ev_{\Gamma}[u_a] \geq \frac{13}{18} - \frac{5}{18}\frac{\prob[u|A]}{n}
		\geq
		\frac{1}{2}
		\]
		
		So following Algorithm~\ref{alg:coloring_ring} gives $a$ an expected utility of at least $\frac{1}{2}$, and any deviation has an expected utility of at most  $\max \{\frac{1}{2}, \ev_{\Gamma}[u_a]\}$, so Algorithm~\ref{alg:coloring_ring} is an equilibrium.

		It remains to prove the $4$ claims above.

		\begin{claim}[1]
			For any agent $a$ and any strategy $T$ that does not include duplication, $\ev_{\Gamma}[u_a] \geq \ev_{T}[u_a]$. In other words, no agent has an incentive to deviate without duplication.
		\end{claim}
		
		\begin{proof}
			Consider the steps of the algorithm from end to start. Let $a$ be an arbitrary agent. If $a$ chooses a color at step $9$, then none of its neighbors has yet to set its color. So $a$ will set $p_a$ and get utility $1$ as decreed by the protocol, so it has no incentive to deviate. If it outputs at step $10-12$, then if any of its neighbors output $p_a$, $a$'s utility is already guaranteed to be $0$ and it has no incentive to deviate. If none of the agents set $p_a$ as their output, $a$ can set it by following the protocol, getting a utility of $1$.
			Then by step $9-12$ there is no incentive to deviate. $a$ can only increase its expected utility by ensuring it selects a color \emph{before} its neighbors. 
			
			At step $7$, in the coin flip process, $a$ knows the size of its consecutive group of agents and its position in it according to the orientation. But by our assumption, the shared coin flip is an equilibrium when agents' utility is directly decided by the result of the coin flip. So there is no deviating strategy for $a$ that increases its expected utility.
			
			At step $6$, $a$ sends $w_a$ before learning any other $w_i$ values, and is thus committed to it and cannot affect the outcome. Furthermore, any deviation regarding other agents' values is caught in $6.c.$ when $\vec{w}$ is validated.
			
			When sharing the preferences, agent $a$ can either send $p_a$ or $p \neq p_a$. any other behavior is detected and the detecting agent would abort. If $a$ sends $p \neq p_a$ and one of its neighbors prefers $p_a$, that neighbor may choose $p_a$ regardless of the shared toss later (it does not know of $a$'s preference). So $a$ increases the probability of an illegal coloring. If that neighbor is part of a larger group of consecutive agents that prefer $p_a$, then $a$ changed nothing: The neighbor would output $p_a$ with probability $\frac{1}{2}$ (according to the shared coin later). If no neighbor prefers $p_a$, $a$ would surely output $p_a$ by sending its preference. Thus it has no incentive to send $p \neq p_a$.
			
			In steps $1-2$, agents share $id$s and decide on an orientation. The expected utility is invariant to the resulting orientation from these steps, as the probability of being the sink and the probability of "winning" the coin flip are identical for any orientation and $id$s.	
		\end{proof}
		
		\begin{claim}[2]
			For any agent $a$ and any duplicating strategy $D$, let $t$ be the round in which $a$ learns $n$. Then one of the following holds:
			\begin{enumerate}
				\item $a$ had already stopped duplicating at round $t$, i.e., $d < n$.
				\item $a$ had duplicated at least $n$ agents.
			\end{enumerate}
		\end{claim}
		
		\begin{proof}
			At the first step of the algorithm, agents send their $id$s to both sides of the ring and forward them. 
			When $a$ duplicates into more than one agent, it has two connecting points, one on each edge of the group of duplicated agents. At every round, $a$ sends two fake $id$s, one across each frontier, and receives two real $id$s. 
			
			Assume $a$ duplicates until it learns $n$.
			In the first round, it sends $2$ $id$s and receives $2$. This goes on until it receives an $id$ it had already received (in case $n$ is even, it receives the same $id$ on both frontiers).  Let $t$ be the round that this occurs. At round $t$, $a$ receives $1$ $id$ it has not already learned (if $n$ is even) or $0$ (if $n$ is odd). But at round $t-1$, $a$ did not know this is about to happen: So $a$ still sent $2$ new "duplicate" $id$s. Thus it created a total of $2t$ duplicate agents, when there are $2t-1$ or $2t-2$ honest agents. The number of honest agents is $n-1$, so $2t > n-1$ and thus $d \geq n$. 
		\end{proof}
		
		\begin{remark}
			This shows that there is a strategy for $a$ in which it duplicates exactly $n+1$ agents: It duplicates until it learns $n$ and then adds one additional agent. However, consider all possible values for $n$: This strategy is beneficial in exactly the same cases as $d = \lfloor \frac{\beta}{2} \rfloor + 1$, and fails in the same cases. Thus it is in fact equivalent or worse, depending on the size of the ring.
		\end{remark}
		
		\begin{claim}[3]
			If $d > n$ then the probability that $a$ is detected and the algorithm is aborted is at least $\frac{1}{2}$.
		\end{claim}
		
		\begin{proof}
			Assume in contradiction that if $d > n$ the probability that $a$ is detected and the algorithm is aborted is strictly less then $\frac{1}{2}$. 
			This probability is only dependent on $\sigma$. But if $\sigma$ is such that the probability of detection is less than $\frac{1}{2}$, then in the same ring running Algorithm~\ref{alg_ks} for $2$-Knowledge Sharing, agents have an incentive to duplicate, contradicting our assumption that Algorithm~\ref{alg_ks} is an equilibrium when agents prior on $n$ is $\sigma$.
		\end{proof}
		\begin{claim}[4]
			If $d \leq n$ under duplication strategy $D$, then $\ev_{\Gamma}[u_a] \geq \ev_{D}[u_a]$, i.e., the expected utility from the duplicating strategy is at most the expected utility for $a$ from following $\Gamma$.
		\end{claim}
		Assume $a$ duplicated into $d \leq n$ agents. First, consider the case where $n$ is odd and all agents prefer $p_a$. 
		
		The only increase in $a$'s expected utility occurs if it increases the probability that its neighbors, denoted $a^-, a^+$ choose a color at round $10-12$ and not at round $9$. If at least one chooses a color at round $9$, it would choose $p_a$ and $a$ gets a utility of $0$. The steps in which agents choose a color are decided by the coin flip $b$ and the identity of the sink. 
		
		We look at the following cases:
		\begin{itemize}
			\item $d$ is odd and neither of $a^-, a^+$ are the sink. In this case, the parity of $a^-, a^+$ is different. Thus, for any coin flip result, at least one of $a^-, a^+$ choose a color at step $9$ and $a$'s utility is $0$. 
			\item $d$ is even, or $d$ is odd and one of $a^-, a^+$ is the sink. assume w.l.o.g. that $a^-$ is the sink. Since $d \leq n$, $a$ cannot affect the outcome of the coin flip, and so the probability that $a^+$ loses the coin flip (and thus $a$ can output $p_a$) is $\frac{1}{2}$. But this is the same expected utility as $a$ would have had it not duplicated at all.
		\end{itemize} 
	
		Notice that even if one of $a$'s neighbors is the sink, $a$'s expected utility does not change! i.e., For any $s$, the expected utility for a is at most $\frac{1}{2}$.
		
		Now, consider the case where there is no sink selection, i.e., the ring is even or not all agents prefer the same color $p_a$. Notice that since there are two frontiers, $a$ can separate the groups of consecutive agents that prefer $p_a$ using the preferences of its duplicates, but still its only effect is to decide whether the parity of $a^+, a^-$ in their respective $C$ groups is the same. Again, if $d$ is even, the parity of $a$'s neighbors is different. If both neighbors want $p_a$, $a$ surely loses because at least one neighbor declares $p_a$ at step $9$, and a cannot output $p_a$ (as a single agent). For each neighbor that does not want $p_a$, $a$'s expected utility is unaffected by its parity. So $a$ would surely duplicate to an even $d$. But still, since $d \leq n$, $a$ cannot affect the coin flip and thus the probability that it wins is exactly $\frac{1}{2}$. That is unless $a^-, a^+$ both prefer $p \neq p_a$, but then $a$ would surely get $p_a$ regardless of the duplication - its utility is $1$.
	\end{proof}


\subsection{Leader Election}
\label{class_LE}

In the Leader Election problem, each agent $a$ outputs $o_a \in \{0,1\}$,
where $o_a = 1$ means that $a$ was elected leader and $o_a = 0$ means otherwise.
$\Theta_L = \{O | \exists a: o_a = 1, \forall{b \neq a}: o_b = 0\}$.
An agent prefers either $0$ or $1$.

Recall that any Leader Election algorithm must be \emph{fair} \cite{DISC13/ADH}, i.e.,
every agent must have equal probability of being elected leader for the algorithm to be an equilibrium.

\begin{theorem}
	\label{theorem:le_uniform}
	In a graph where agents' prior on $n$ is uniform, an equilibrium exists if and only if $\beta \leq \alpha + 1$
\end{theorem}

\begin{proof}
	If $\beta \leq \alpha + 1$, we show a \emph{fair} leader election protocol (as shown in \cite{Afek:2014:DCB:2611462.2611481,Abraham:2011}) is an equilibrium. In any fair leader election protocol any agent has a probability of $\frac{1}{n'}$ to be elected. If an agent follows the protocol, its expected utility is $\frac{1}{n}$, i.e., it is $\frac{1}{2}\frac{1}{\alpha} + \frac{1}{2}\frac{1}{\alpha + 1}$.
	
	If an agent duplicates into more than $2$ agents ($d>2$) it is surely caught.
	If an agent duplicates into $2$ agents (i.e., $d=2$) then $n' = n+1$.
	This leads to two equally possible scenarios:
	\begin{itemize}
		\item $n = \alpha$,
		thus the probability that one of $a$'s duplicates is elected leader is $\frac{2}{\alpha+1}$.
		\item $n = \alpha + 1$, thus $n' = \alpha + 2 > \beta$
		and the deviation is detected, resulting in a utility of $0$.
	\end{itemize}
	Overall, the duplication strategy has an expected utility of
	$\frac{1}{2} \cdot \frac{2}{\alpha+1} = \frac{1}{\alpha + 1}$.
	Thus, duplicating into $2$ agents does not increase $a$'s utility.
	
	When $\beta > \alpha + 1$, we show that any fair leader election is \emph{not} an equilibrium.
	Note that the expected utility from following a fair leader election protocol is:
	\[
		\ev_{n \sim \sigma} [\frac{1}{n}] 
		= 
		\frac{1}{\beta-\alpha + 1} \cdot \sum_{\alpha}^{\beta} \frac{1}{n}
	\]
	
	Now consider $\beta = \alpha + 2$. If $a$ duplicates into $d = 2$ agents, consider its expected utility for every possible $n$:
	\begin{itemize}
		\item If $n = \alpha$, $a$ enters the fair leader election with $2$ possible winning candidates out of $\alpha + 1$, and thus its expected utility is $\frac{2}{\alpha + 1}$.
		\item If $n = \alpha + 1$, $a$'s expected utility is $\frac{2}{\alpha + 2}$
		\item If $n = \alpha + 2 = \beta$, $a$'s duplication is detected and its utility is $0$. 
	\end{itemize}
	Overall, this yields an expected utility of $\frac{2}{3} (\frac{1}{\alpha + 1} + \frac{1}{\alpha + 2})$, greater than the legal utility for $\beta = \alpha +2$ for any $\alpha$. 
	
	For any increase in $\beta$, we have another case where the cheater is not detected and earns an expected utility of $\frac{2}{n'}$. Thus the cheater's expected utility is:
	\[
		\ev_{D}[u_a] = \frac{\beta-\alpha}{\beta-\alpha + 1} 
		\big(
			\frac{2}{\alpha+ 1} + \dots + \frac{2}{\beta}
		\big)
		=
		2 \frac {\beta - \alpha}{\beta - \alpha + 1} \cdot \sum_{n'=\alpha + 1}^{\beta} \frac{1}{n'}		
	\]
	
	For any $\beta > \alpha + 1$, this is greater than the expected utility from following the protocol.
	
\end{proof}

\begin{theorem}
	\label{theorem:le_geometric}
	In a graph where agents' prior on $n$ is geometric, an equilibrium exists if and only if $\beta = \alpha$.
\end{theorem}
\begin{proof}
	For any $\beta > \alpha + 1$, the only difference from the proof of Theorem~\ref{theorem:le_uniform} is that now all the cases in which the cheater gains have a higher probability, thus $a$ has an even stronger incentive to duplicate $d=2$ and increase its expected utility. It remains to show that for $\beta = \alpha + 1$, $a$ has an incentive to duplicate $d=2$. 
	
	Let $\beta = \alpha + 1$, and denote $D$ the strategy in which a cheating agent $a$ duplicates into $d = 2$ agents. If $n = \alpha$, its expected utility is $\frac{2}{\alpha + 1}$, otherwise its utility is $0$. Then we know that:
	\[
		\ev_D[u_a] = \prob[n = \alpha] \cdot \frac{2}{\alpha+1}= (\frac{1}{2} + \frac{1}{8}) \frac{2}{\alpha+1} = \frac{5}{4} \big( \frac{1}{\alpha + 1} \big)
	\]
	On the other hand, the expected utility from following the protocol is $\frac{1}{2}\big( \frac{1}{\alpha} + \frac{1}{\alpha + 1} \big)$. For any $\alpha > 2$, the expected utility from $D$ is strictly greater than from following the protocol.
\end{proof}

\subsection{Ring Partition}
\label{class_partition}

In the Ring Partition problem, the agents of an even-sized ring are partitioned
into two, equally-sized groups: group $0$ and group $1$.
An agent prefers to belong to either group $0$ or $1$.

\begin{theorem}
	\label{theorem:partition-bound}
	There is an equilibrium for Ring Partition for any $\sigma$ with any support, including an infinite support.
\end{theorem}

\begin{proof}
	It is clear that an agent will not duplicate itself to change the parity of the graph,
	as that will necessarily cause an erroneous output.
	So it is enough to show an algorithm that is an equilibrium
	for even graphs, when agents have no knowledge about $n$.
	Consider the following algorithm:
	\begin{itemize}
		\item Either one arbitrary agent wakes up or we run a Wake-Up subroutine and then
		Leader Election \cite{Afek:2014:DCB:2611462.2611481}.
		Since the initiator (leader) has no effect on the output,
		both are an equilibrium in this application.
		\item The initiating agent sends a token which alternatively marks agents by 0 or 1 and also defines the direction of communication in the ring.
		\item Upon reception of the token with value $m$, an agent $a$ does one of the following:
		\begin{enumerate}
			\item If $m = 0$, send predecessor (denoted $p(a)$) a random bit $t_a$.
			\item Else, if $m = 1$, wait for 1 round and send successor (denoted $s(a)$) a random bit $t_a$.
		\end{enumerate}
		\item Upon reception of the neighbor's bit (one round after receiving the token), set 
		$$ o_a = (t_a + t_{s(a)/p(a)} + m)_{mod2} $$
		\item As the token arrives back at the initiator, it checks the token's parity. For even rings, it must be the opposite value from the value it originally sent.
	\end{itemize}
	
	This algorithm divides every pair of agents to one with output $1$ and one with output $0$, as the token value $m$ is different, thus achieving a partition.
	
	We show that it is also an equilibrium.	
	Assume an agent $a$ deviates at some round $r$. If $r$ is in the Wake-Up or Leader Election
	phase in order to be the initiator, it cannot improve its utility since choosing the starting value of the token, choosing the direction, or being first cannot increase the agent's utility. 
	If it is a deviation while the token traverses other parts of the graph, any message $a$ sends will eventually be discovered, as the real token has either already passed or will eventually pass through the "cheated" agent.
	If $a$ changes the value of the token, a randomization between two agents will be out of sync eventually at the end of the token traversal, and also the initiator will recognize that the ring does not contain an even number of agents.
	During the exchange of $t_a$ the result is independent of $a$'s choice of value for $t_a$.
	So there is no round in which $a$ can deviate from the protocol.
	
	Since agents will only duplicate into $d$ that maintains the parity of the ring, and since the duplicates' output  must be consistent with their neighbors' without the ability to affect a coin toss over a single edge, a cheater has no incentive to deviate by duplicating.
\end{proof}

\subsection{Orientation}
\label{class_orientation}

In the Orientation problem, the two ends of each edge must agree on a direction for the edge.
An agent prefers certain directions for its edges.


\begin{theorem}
	\label{theorem:orientation-bound}
	There is an equilibrium for Orientation for any $\sigma$ with any support, including an infinite support.
\end{theorem}

\begin{proof}
	We show a simple algorithm and prove that it is an equilibrium for any prior $\sigma$.
	Assuming all agents start at the same round (otherwise run Wake-Up), consider the following algorithm:
	\begin{itemize}
		\item Each agent simultaneously send a random number $(0,1)$ and its $id$ on each of its edges.
		\item For each edge, XOR the bit you sent and the bit received over that edge
		\item If the result is 1, the edge is directed towards the agent with the higher $id$, otherwise it is directed towards the lower $id$.
		\item Every agent outputs the list of pairs with $id$ and direction for each of its neighbors. 
	\end{itemize}
	
	Since an agent's utility is defined over its personal output, Solution Preference inhibits agents to output a correct set of pairs, so a cheater may only influence the direction of the edges.
	Since duplication does not create any new edges between the cheater and the original graph, and the orientation is decided over each edge independently, it does not effect any agent's utility. Other than that, randomizing a single bit over an edge at the same round is in equilibrium. So the algorithm is an equilibrium regardless of $\sigma$.
\end{proof}

\section{Discussion}
\label{section_discussion}

In this paper we have shown that the assumption that $n$ is a-priori known,
commonly made in previous works, has a critical role in the possibility of equilibrium.
In realistic scenarios, the exact size of the network may not be known to all members,
or only estimates on the exact size are known in advance.
In such networks, the use of duplication gives an agent power
to affect the outcome of most algorithms, and in some cases makes equilibrium impossible.

When there is a finite prior for $n$,
our proofs in Section~\ref{section_approx} show that the initial knowledge required for equilibrium
depends on the balance between two factors:
the amount of duplications necessary to increase
an agent's expected utility,
and the probability that the cheater is caught duplicating.
In order for an agent to have an incentive to duplicate itself,
an undetected duplication needs to be more profitable than following the algorithm
while also involving low risk of being caught.

Our results suggest several open directions that may be of interest:
\begin{enumerate}
	\item Finding an equilibrium for Knowledge Sharing in a general graph when agents duplicate into any $d \leq n$. This would be the missing piece that, along with our impossibility proof in Theorem~\ref{theorem:KSgeneral},
	would prove the bounds in Table~\ref{table:ks_bounds} hold for general graphs.

	\item Algorithms and impossibility results for other problems,
	as well as bounds for various priors.
	
	
	
	\item Finding more problems that have equilibrium when the prior of $n$ has infinite support,
	in any graph (unlike Partition)
	and a non-constant radius in the LOCAL model (unlike Orientation).
	
	\item Exploring the effects of initial knowledge of $n$
	in an asynchronous setting.
\end{enumerate}

\section{Acknowledgment}
We would like to thank Doron Mukhtar for showing us
the ring partition problem and showing an equilibrium for it for any $\sigma$,
when we thought such problems do not exist,
and to Michal Feldman and Amos Fiat for helpful discussions.

\clearpage
\bibliographystyle{abbrv}

\begin{thebibliography}{10}

\bibitem{Abraham:2011}
I.~Abraham, L.~Alvisi, and J.~Y. Halpern.
\newblock Distributed computing meets game theory: Combining insights from two
  fields.
\newblock {\em SIGACT News}, 42(2):69--76, June 2011.

\bibitem{PODC/AbrahamDGH06}
I.~Abraham, D.~Dolev, R.~Gonen, and J.~Y. Halpern.
\newblock Distributed computing meets game theory: robust mechanisms for
  rational secret sharing and multiparty computation.
\newblock In {\em PODC}, pages 53--62, 2006.

\bibitem{TCC/AbrahamDH08}
I.~Abraham, D.~Dolev, and J.~Y. Halpern.
\newblock Lower bounds on implementing robust and resilient mediators.
\newblock In {\em TCC}, pages 302--319, 2008.

\bibitem{DISC13/ADH}
I.~Abraham, D.~Dolev, and J.~Y. Halpern.
\newblock Distributed protocols for leader election: A game-theoretic
  perspective.
\newblock In {\em DISC}, pages 61--75, 2013.

\bibitem{Afek:2014:DCB:2611462.2611481}
Y.~Afek, Y.~Ginzberg, S.~Landau~Feibish, and M.~Sulamy.
\newblock Distributed computing building blocks for rational agents.
\newblock In {\em Proceedings of the 2014 ACM Symposium on Principles of
  Distributed Computing}, PODC '14.

\bibitem{fullpaper}
Y.~{Afek}, S.~{Rafaeli}, and M.~{Sulamy}.
\newblock {Cheating by Duplication: Equilibrium Requires Global Knowledge}.
\newblock {\em ArXiv e-prints}, Nov. 2017.

\bibitem{afek_et_al:LIPIcs:2018:9794}
Y.~Afek, S.~Rafaeli, and M.~Sulamy.
\newblock {The Role of A-priori Information in Networks of Rational Agents}.
\newblock In U.~Schmid and J.~Widder, editors, {\em 32nd International
  Symposium on Distributed Computing (DISC 2018)}, volume 121 of {\em Leibniz
  International Proceedings in Informatics (LIPIcs)}, pages 5:1--5:18,
  Dagstuhl, Germany, 2018. Schloss Dagstuhl--Leibniz-Zentrum fuer Informatik.

\bibitem{SOSP/AiyerACDMP05}
A.~S. Aiyer, L.~Alvisi, A.~Clement, M.~Dahlin, J.-P. Martin, and C.~Porth.
\newblock Bar fault tolerance for cooperative services.
\newblock In {\em SOSP}, pages 45--58, 2005.

\bibitem{AttiyaBook}
H.~Attiya and J.~Welch.
\newblock {\em Distributed Computing: Fundamentals, Simulations and Advanced
  Topics}.
\newblock John Wiley \& Sons, 2004.

\bibitem{Awerbuch:1989:NDL:1398514.1398717}
B.~Awerbuch, M.~Luby, A.~V. Goldberg, and S.~A. Plotkin.
\newblock Network decomposition and locality in distributed computation.
\newblock In {\em Proceedings of the 30th Annual Symposium on Foundations of
  Computer Science}, SFCS '89, pages 364--369, Washington, DC, USA, 1989. IEEE
  Computer Society.

\bibitem{gt_idspace}
D.~{Bank}, M.~{Sulamy}, and E.~{Waserman}.
\newblock {Reaching Distributed Equilibrium with Limited ID Space}.
\newblock {\em ArXiv e-prints}, Apr. 2018.

\bibitem{Barany1992}
I.~B\'{a}r\'{a}ny.
\newblock Fair distribution protocols or how the players replace fortune.
\newblock {\em Math. Oper. Res.}, 17(2):327--340, May 1992.

\bibitem{JET/Ben-Porath03}
E.~Ben-Porath.
\newblock Cheap talk in games with incomplete information.
\newblock {\em J. Economic Theory}, 108(1):45--71, 2003.

\bibitem{Bhattacharjee:2005:ABS:1080192.1080203}
R.~Bhattacharjee and A.~Goel.
\newblock Avoiding ballot stuffing in ebay-like reputation systems.
\newblock In {\em Proceedings of the 2005 ACM SIGCOMM Workshop on Economics of
  Peer-to-peer Systems}, P2PECON '05, pages 133--137, New York, NY, USA, 2005.
  ACM.

\bibitem{Bianchini:2005:IP:1052934.1052938}
M.~Bianchini, M.~Gori, and F.~Scarselli.
\newblock Inside pagerank.
\newblock {\em ACM Trans. Internet Technol.}, 5(1):92--128, Feb. 2005.

\bibitem{Cheng:2005:SRM:1080192.1080202}
A.~Cheng and E.~Friedman.
\newblock Sybilproof reputation mechanisms.
\newblock In {\em Proceedings of the 2005 ACM SIGCOMM Workshop on Economics of
  Peer-to-peer Systems}, P2PECON '05, pages 128--132, New York, NY, USA, 2005.
  ACM.

\bibitem{Cole:1986:DCT:10366.10368}
R.~Cole and U.~Vishkin.
\newblock Deterministic coin tossing with applications to optimal parallel list
  ranking.
\newblock {\em Inf. Control}, 70(1):32--53, July 1986.

\bibitem{PODC/DaniMRS11}
V.~Dani, M.~Movahedi, Y.~Rodriguez, and J.~Saia.
\newblock Scalable rational secret sharing.
\newblock In {\em PODC}, pages 187--196, 2011.

\bibitem{CRYPTO/DodisHR00}
Y.~Dodis, S.~Halevi, and T.~Rabin.
\newblock A cryptographic solution to a game theoretic problem.
\newblock In {\em CRYPTO}, pages 112--130, 2000.

\bibitem{Douceur:2002:SA:646334.687813}
J.~R. Douceur.
\newblock The sybil attack.
\newblock In {\em Revised Papers from the First International Workshop on
  Peer-to-Peer Systems}, IPTPS '01, pages 251--260, London, UK, UK, 2002.
  Springer-Verlag.

\bibitem{TCC/FuchsbauerKN10}
G.~Fuchsbauer, J.~Katz, and D.~Naccache.
\newblock Efficient rational secret sharing in standard communication networks.
\newblock In {\em TCC}, pages 419--436, 2010.

\bibitem{SCN/GordonK06}
S.~D. Gordon and J.~Katz.
\newblock Rational secret sharing, revisited.
\newblock In {\em SCN}, pages 229--241, 2006.

\bibitem{ICALP/GroceKTZ12}
A.~Groce, J.~Katz, A.~Thiruvengadam, and V.~Zikas.
\newblock Byzantine agreement with a rational adversary.
\newblock In {\em ICALP (2)}, pages 561--572, 2012.

\bibitem{Halpern:2016:RCE:2933057.2933088}
J.~Y. Halpern and X.~Vila\c{c}a.
\newblock Rational consensus: Extended abstract.
\newblock In {\em Proceedings of the 2016 ACM Symposium on Principles of
  Distributed Computing}, PODC '16, pages 137--146, New York, NY, USA, 2016.
  ACM.

\bibitem{Kuhn:2006:CDG:1146381.1146387}
F.~Kuhn and R.~Wattenhofer.
\newblock On the complexity of distributed graph coloring.
\newblock In {\em Proceedings of the Twenty-fifth Annual ACM Symposium on
  Principles of Distributed Computing}, PODC '06, pages 7--15, New York, NY,
  USA, 2006. ACM.

\bibitem{PODC/LepinskiMP04}
M.~Lepinski, S.~Micali, C.~Peikert, and A.~Shelat.
\newblock Completely fair sfe and coalition-safe cheap talk.
\newblock In {\em PODC}, pages 1--10, 2004.

\bibitem{Linial1986}
N.~Linial.
\newblock Legal coloring of graphs.
\newblock {\em Combinatorica}, 6(1):49--54, 1986.

\bibitem{Linial:1987:DGA:1382440.1382990}
N.~Linial.
\newblock Distributive graph algorithms global solutions from local data.
\newblock In {\em Proceedings of the 28th Annual Symposium on Foundations of
  Computer Science}, SFCS '87, pages 331--335, Washington, DC, USA, 1987. IEEE
  Computer Society.

\bibitem{CYPTO/LysyanskayaT06}
A.~Lysyanskaya and N.~Triandopoulos.
\newblock Rationality and adversarial behavior in multi-party computation.
\newblock In {\em CRYPTO}, pages 180--197, 2006.

\bibitem{TARK/McGrewPS03}
R.~McGrew, R.~Porter, and Y.~Shoham.
\newblock Towards a general theory of non-cooperative computation.
\newblock In {\em TARK}, pages 59--71, 2003.

\bibitem{PODC/MoscibrodaSW06}
T.~Moscibroda, S.~Schmid, and R.~Wattenhofer.
\newblock When selfish meets evil: byzantine players in a virus inoculation
  game.
\newblock In {\em PODC}, pages 35--44, 2006.

\bibitem{CACM/Shamir79}
A.~Shamir.
\newblock How to share a secret.
\newblock {\em Commun. ACM}, 22(11):612--613, 1979.

\bibitem{TCS/Shoham05}
Y.~Shoham and M.~Tennenholtz.
\newblock Non-cooperative computation: Boolean functions with correctness and
  exclusivity.
\newblock {\em Theoretical Computer Science}, 343(1–2):97 -- 113, 2005.

\bibitem{Szegedy:1993:LBG:167088.167156}
M.~Szegedy and S.~Vishwanathan.
\newblock Locality based graph coloring.
\newblock In {\em Proceedings of the Twenty-fifth Annual ACM Symposium on
  Theory of Computing}, STOC '93, pages 201--207, New York, NY, USA, 1993. ACM.

\bibitem{ECON/UAV02}
A.~Urbano and J.~E. Vila.
\newblock Computational complexity and communication: Coordination in
  two-player games.
\newblock {\em Econometrica}, 70(5):1893--1927, September 2002.

\bibitem{ET/UA04}
A.~Urbano and J.~E. Vila.
\newblock Computationally restricted unmediated talk under incomplete
  information.
\newblock {\em Economic theory}, 2004.

\bibitem{OPODIS/WongLACD11}
E.~L. Wong, I.~Levy, L.~Alvisi, A.~Clement, and M.~Dahlin.
\newblock Regret freedom isn't free.
\newblock In {\em OPODIS}, pages 80--95, 2011.

\bibitem{YifrachLeader}
A.~Yifrach and Y.~Mansour.
\newblock {\em Fair Leader Election for Rational Agents}.
\newblock PODC '18.

\end{thebibliography}

\end{document}